\documentclass[10pt]{article}
\usepackage{graphicx,epsfig,amsmath,amsthm,amssymb}
\usepackage{enumerate}
\usepackage{multirow}
\usepackage{bbm}
\usepackage{balance}
\usepackage{footnote}
\usepackage{tablefootnote}
\usepackage{xr}
\newcommand{\ex}[1]{\mathbb{E}\left[ #1 \right] }
\newcommand{\norm}[1]{\left\lVert #1 \right\rVert}
\newcommand{\Q}{ {\mathbf Q}}

\newcommand{\Qc}{ {\mathbf Q}_{\perp}}
\newcommand{\Qp}{ {\mathbf Q}_{\parallel}}
\newcommand{\A}{ {\mathbf A}}
\newcommand{\s}{ {\mathbf S}}
\newcommand{\UU}{ {\mathbf U}}
\newcommand{\inner}[2]{\langle #1, #2 \rangle}
\newcommand{\lep}[1]{\mathop  \le \limits^{(#1)}}
\newcommand{\gep}[1]{\mathop  \ge \limits^{(#1)}}
\newcommand{\gp}[1]{\mathop  > \limits^{(#1)}}
\newcommand{\ep}[1]{\mathop  = \limits^{(#1)}}

\newtheorem{definition}{Definition}
\newtheorem{lemma}{Lemma}
\newtheorem{claim}{Claim}
\newtheorem{proposition}{Proposition}

\newtheorem{theorem}{Theorem}

\newtheorem{remark}{Remark}

\begin{document}

%\setcopyright{cagovmixed}

% \doi{10.475/123_4}

% % ISBN
% \isbn{123-4567-24-567/08/06}

% %Conference
% \conferenceinfo{PLDI '13}{June 16--19, 2013, Seattle, WA, USA}

% \acmPrice{\$15.00}

%
% --- Author Metadata here ---
%\conferenceinfo{WOODSTOCK}{'97 El Paso, Texas USA}
%\CopyrightYear{2007} % Allows default copyright year (20XX) to be over-ridden - IF NEED BE.
%\crdata{0-12345-67-8/90/01}  % Allows default copyright data (0-89791-88-6/97/05) to be over-ridden - IF NEED BE.
% --- End of Author Metadata ---

\title{Flexible Load Balancing with Multi-dimensional State-space Collapse: Throughput and Heavy-traffic Delay Optimality}
\author{Xingyu Zhou\\Department of ECE\\The Ohio State University\\zhou.2055@osu.edu \and Jian Tan\\Department of ECE\\The Ohio State University\\tan.252@osu.edu\\
\and Ness Shroff\\Department of ECE and CSE\\The Ohio State University\\
shroff.11@osu.edu } 
\date{}
\maketitle

\begin{abstract}

Heavy traffic analysis for load balancing policies has relied heavily on the condition of state-space collapse onto a single-dimensional line in previous works. 
In this paper, via Lyapunov-drift analysis, we rigorously prove that even under a multi-dimensional state-space collapse, steady-state heavy-traffic delay optimality can still be achieved for a general load balancing system.
This result directly implies that achieving steady-state heavy-traffic delay optimality simply requires that no server is idling while others are busy at heavy loads, thus complementing and extending the result obtained by diffusion approximations.
Further, we explore the greater flexibility provided by allowing a multi-dimensional state-space collapse in designing new load balancing policies that are both throughput optimal and heavy-traffic delay optimal in steady state.
This is achieved by overcoming various technical challenges, and the methods used in this paper could be of independent interest. 

\end{abstract}

% \keywords{Heavy-traffic delay optimality; Throughput optimality; Flexible load balancing; Multi-dimensional state-space collapse}

% \category{G.3}{Probability and Statistics}{Stochastic processes,
% Queueing theory}

% !TEX root = ./performance18.tex
\section{Introduction}

We consider a discrete-time load balancing system which consists of one dispatcher and $N$ servers, each associated with an infinite buffer queue. The service rate of server $n$ is $\mu_n$. At each time-slot $t$, the exogenous tasks arrive with rate $\lambda_{\Sigma}$, and upon arrival each task is immediately dispatched to one of the queues. A load balancing policy is a rule that selects the queue to which a new arrival in each time-slot should be dispatched. 
In recent years the development of efficient load balancing policies has received significant attention because of their applicability in distributed architectures such as Web service~\cite{gupta2007analysis}, large data storage systems (e.g., HBase~\cite{george2011hbase}), cloud computing systems~\cite{foster2008cloud}, etc.
% In order to improve the overall system performance, the question of how to design efficient load balancing policies has attracted strong interest in the last several years, motivated by the significant challenges of task assignments in distributed architectures such as Web service~\cite{gupta2007analysis}, large data storage systems (e.g., HBase~\cite{george2011hbase}), cloud computing systems~\cite{foster2008cloud}, etc. 
A desirable load balancing policy is often one that is able to improve the average response time while achieving high utilization of resources. 
To this end, many works in the literature have focused on minimizing the average delay in the heavy-traffic regime where the exogenous arrival rate approaches the boundary of the capacity region, i.e., the heavy-traffic parameter $\epsilon = \sum \mu_n - \lambda_{\Sigma}$ approaches zero in our system.

At the heart of most heavy-traffic analysis is the notion of \textit{state-space collapse}, which roughly means that the original multi-dimensional system space concentrates around a single dimensional (or generally a lower dimensional) subspace as the heavy-traffic parameter $\epsilon$ goes to zero. For instance, via either diffusion approximations~\cite{foschini1978basic} or the recently developed drift-based framework~\cite{eryilmaz2012asymptotically}, it has been shown that under the so-called join-shortest-queue (JSQ) policy, the load balancing system in heavy-traffic would collapse to a one-dimensional line where all the queue lengths are equal. This indicates that the system behaves as if there is only a single queue with all the servers pooled together as an aggregated server, which is often called $\textit{complete resource pooling}$. This result directly implies that JSQ is asymptotically optimal, i.e., heavy-traffic delay optimal, since the response time in the pooled single-server system is stochastically less than that of a typical load balancing system. The same one-dimensional state-space collapse is also the key in establishing heavy-traffic delay optimality for the so-called power-of-$d$ policy~\cite{chen2012asymptotic},~\cite{maguluri2014heavy}, where the dispatcher routes the new arrival to a server with the shortest queue length among $d \ge 2$ servers selected uniformly at random. Instead of requiring the information of all the queue lengths as in JSQ, the power-of-$d$ policy is able to achieve asymptotic optimality with partial queue length information, and hence provides greater flexibility and scalability for large-scale distributed systems. 

% As illustrated by the above, one could naturally think that in order to be heavy-traffic delay optimal, a load balancing policy should always guarantee that all the queue lengths are equal in the heavy-traffic limit, i.e., one-dimensional state-space collapse. 
The authors in \cite{kelly1993dynamic} argue heuristically that state-space collapse to a line where all the queue lengths are equal may not be necessary for showing heavy-traffic delay optimality in load balancing systems. In particular, they proposed a symmetric threshold policy for a load balancing system with \emph{two} homogeneous servers, and conjectured that as long as the threshold satisfies a certain property, the total work process under the threshold policy has the same diffusion limit as that under the optimal JSQ. Nevertheless, the system state in heavy-traffic limit under the threshold policy is now in the two-dimensional positive orthant rather than a single-dimensional line where all the queue lengths are equal. Hence, the authors in \cite{kelly1993dynamic} argued that the key feature of a heavy-traffic optimal policy is to keep all the servers busy when there is substantial work rather than the strong property of maintaining all the queue lengths equal. This argument is validated in a two-server system with an asymmetric threshold policy proposed in \cite{teh2002critical}, under which the total work process is proven to have the same diffusion limit as that of JSQ while the state space collapses to a \textit{line} where the lengths of two queues are not equal. 
Note that besides only considering a two-server system, a further limitation  in both \cite{kelly1993dynamic} and \cite{teh2002critical} is that the asymptotic optimality holds only for a finite time interval. This is because the interchange of limits has not been established for the diffusion approximations. 
% It is worth pointing out that besides the special two-server case under a particular policy that requires the knowledge of arrival rate, the asymptotic optimality in both \cite{kelly1993dynamic} and \cite{teh2002critical} only holds for a finite time interval as a result of not establishing the interchange of limits for the diffusion approximations. 
Therefore, an interesting open problem is whether or not \textit{steady-state} delay optimality in heavy-traffic holds under a \emph{multi-dimensional} state-space collapse for a \emph{general} load balancing system, and if so, how one can design a load balancing policy to achieve it.

In this paper, we take a systematic approach to addressing this problem. First, we extend the recently developed drift-based framework in \cite{eryilmaz2012asymptotically} to rigorously show that even under a multi-dimensional state-space collapse, a load balancing policy is still able to achieve heavy-traffic delay optimality in \textit{steady-state}. This result then allows us to explore the \textit{flexibility} in designing load balancing policies that are not only throughput optimal but also heavy-traffic delay optimal in steady-state. 
The main contributions of this paper can be summarized as follows:
\begin{itemize}
\item We rigorously establish heavy-traffic delay optimality in steady-state under a multi-dimensional state-space collapse for a general load balancing system. More precisely, we consider a symmetric finitely generated cone $\mathcal{K}_{\alpha}$ parameterized by a nonnegative $\alpha \in [0,1]$. In particular, when $\alpha = 1$ the cone reduces to the line where all the components are equal, and when $\alpha = 0$ the cone is the nonnegative orthant. Our first main result (cf. Theorem \ref{thm:theorem_1}) states that given a throughput optimal load balancing policy, if the system state collapses to a cone $\mathcal{K}_{\alpha}$ with any fixed $\alpha \in (0,1]$, this policy is heavy-traffic delay optimal in steady-state. The importance of this result is two-fold: (i) it rigorously proves a conjectured insight behind steady-state heavy-traffic delay optimality in load balancing systems. In particular, it shows that to achieve the heavy-traffic optimality in steady-state for a general system,  a load balancing policy should also just be able to keep all the servers busy when there is substantial work, rather than the strong requirement of maintaining all the queue lengths equal. This complements and extends the diffusion approximation results in \cite{kelly1993dynamic},~\cite{teh2002critical}. (ii) it enables us to establish heavy-traffic delay optimality under general state-space collapse regions (including even non-convex regions). This can be achieved by showing that the actual state-space collapse region can be covered by a cone $\mathcal{K}_{\alpha}$ with some $\alpha \in (0,1]$, which directly implies that the system state also collapses to the cone $\mathcal{K}_{\alpha}$, and hence heavy-traffic delay optimality follows from Theorem \ref{thm:theorem_1}.

\item By exploiting the key implications of the first result, we are then able to characterize the degree of flexibility (from two different dimensions) in designing new load balancing policies that are both throughput-optimal and heavy traffic delay-optimal in steady-state. (i) The first dimension of flexibility is concerned with the frequency of favoring shorter queues. We find that instead of favoring shorter queues at each time-slot for every system-state, it is sufficient to favor shorter queues only when the system-state is outside a cone $\mathcal{K}_{\alpha}$ for any fixed $\alpha \in (0,1]$. This means that whenever the system-state is within the cone $\mathcal{K}_{\alpha}$, the dispatcher is allowed to use an arbitrary Markovian dispatching distribution, and this flexibility increases as $\alpha$ approaches zero. 
% Note that the flexibility in the first dimension is extremely useful when there exists data-locality problem in a load balancing system under which it is often impossible to favor shorter queues for all system-states.
(ii) The second dimension is related to the intensity with which shorter queues are favored. We find that instead of only joining the shortest queue as in the JSQ policy or having monotone decreasing probabilities from joining the shortest queue to the longest queue in the power-of-$d$ policy, an even weaker intensity of favoring shorter queues is sufficient and this intensity can be characterized by some parameter $\delta$. 
The above flexibilities from two different dimensions are stitched together in Theorem \ref{thm:theorem_2}. We also consider the case where these two flexibilities scale with the heavy-traffic parameter $\epsilon$, i.e., both $\alpha$ and $\delta$ decrease to zero as $\epsilon$ approaches zero. We show that steady-state heavy-traffic delay optimality is preserved as long as $\alpha \delta = \Omega(\epsilon^{\beta})$ for any $\beta \in [0,1)$ (cf. Proposition \ref{prop:prop_3}). This result offers us even more flexibility in designing efficient load balancing policies.
\item {The techniques used in this paper are of independent interest. For example, in order to establish throughput optimality defined in this paper, namely positive recurrence with bounded moments in steady-state, the standard stochastic drift analysis of a suitable Lyapunov function is very difficult in our case because the drift within the cone $\mathcal{K}_{\alpha}$ can be positive. On the other hand, the standard fluid limit methods can only be used to show positive recurrence but not bounded moments of the queue lengths. To address the problem, we combine fluid approximations with stochastic Lyapunov theory. In particular, we show that the Lyapunov function used in the fluid model, i.e., the sum of the queue lengths, is also a suitable Lyapunov function for the original stochastic system. This connection allows us to carry out the drift analysis in the fluid domain. Then, we come back to the stochastic system and apply the drift-based analysis of the same Lyapunov function in the original system by leveraging the result in the fluid domain to show both positive recurrence and bounded moments. Also, for the result of state-space collapse to a cone, the standard analysis adopted in the single-dimensional state-space collapse fails as well in this case. This is because in contrast to the projection onto a line, the projection onto a convex cone is more complex. 
In fact, a closed-form formula of the projection onto a polyhedral cone is still an open problem. To circumvent this difficulty, instead of obtaining the exact projection, we find that it is sufficient to establish a monotone property of the projection to show that the system state collapses to a cone. Moreover, the bounds on the moments in the state-space collapse result hold even when the system is not in heavy-traffic, and hence can be independently used as a performance evaluation tool for the pre-limit load balancing systems.}

\end{itemize}

\subsection{Related Work}
% A very natural choice of load balancing policies especially for small-scale systems is the Join-Shortest-Queue (JSQ), under which the dispatcher continuously samples all the servers upon new arrivals and dispatch them to the least loaded server in the system. It has been shown that JSQ is delay optimal in the sense of stochastic order for different system settings[][]. Nevertheless, the performance of this policy comes at the cost of substantial message overhead as it has to constantly know the queue lengths of all the servers, which is not feasible in large-scale systems. As a result, an alternative load balancing policy with low message overhead called power-of-$d$ is proposed []. Under this policy, the dispatcher probes $d$ servers uniformly at random and dispatches new arrivals to the server with the shortest queue among the $d$ servers. 

The use of state-space collapse to study the delay performance in the heavy-traffic regime was introduced in~\cite{foschini1978basic} for two parallel servers. The authors, via diffusion approximations, showed that the two separate servers under the JSQ policy act as a pooled resource in the heavy-traffic limit. Since then, the methodology of diffusion limits combined with state-space collapse has been adopted in a number of papers on parallel servers~\cite{reiman1984some}~\cite{bell2001dynamic}~\cite{chen2012asymptotic}~\cite{bramson1998state}~\cite{harrison1998heavy}. For example, the author in~\cite{reiman1984some} generalized the results in~\cite{foschini1978basic} to the case of renewal arrivals and general service times. In~\cite{chen2012asymptotic}, power-of-$d$ was shown to have the same diffusion limit as JSQ in the heavy-traffic limit. The common step in all these works is to show that the diffusion limit in heavy-traffic converges to a one-dimensional Brownian motion, which implies sample-path optimality in finite time. However, in order to show optimality in steady-state, an interchange of limit argument needs to be proven, which is often difficult and not undertaken in the aforementioned works.
Recently, the authors in~\cite{eryilmaz2012asymptotically} proposed a drift-based framework, which is able to directly work on the stationary distribution and establish steady-state heavy-traffic optimality of load balancing policies such as JSQ, and scheduling policies such as MaxWeight. This framework has been adopted to show steady-state heavy-traffic optimality of several policies in different scenarios. For instance, based on this framework, the authors in~\cite{maguluri2014heavy} established the steady-state heavy-traffic optimality of power-of-$d$ policy. {{The authors in~\cite{zhou2017designing} identified a class of heavy-traffic delay optimal policies.}} Moreover, it has been shown in~\cite{wang2016maptask} that a joint JSQ and MaxWeight policy is heavy-traffic delay optimal for MapReduce clusters under a specific traffic scenario. For all traffic scenarios, a heavy-traffic delay optimal policy called `local-task-first' policy was proposed in \cite{xie2015priority} based on this new framework.

However, it is worth noting that the state-space collapse of all the aforementioned heavy-traffic optimal load balancing policies is only one-dimensional. A two-dimensional state-space collapse was considered in \cite{kelly1993dynamic}, in which the authors argued heuristically that heavy-traffic delay optimality is preserved in this case, and hence claimed that the key feature of a heavy-traffic optimal load balancing policy is to keep all the servers busy when there is substantial remaining work, rather than the strong property of maintaining all the queue lengths equal. The authors in \cite{teh2002critical} validated this claim for a two-server system under an asymmetric threshold policy. In particular, they showed that the diffusion limit of the work process is the same as that under JSQ, while the state-space collapses to a \textit{line} where the queue lengths of two servers are not equal. However, both the results in \cite{kelly1993dynamic} and \cite{teh2002critical} hold only for a finite time since the validity of the interchange of limits argument was not established for the diffusion approximations. 
Motivated by this, in this paper, we extend the recently developed drift-based framework \cite{eryilmaz2012asymptotically} and successfully establish \textit{steady-state} heavy-traffic optimality under a multi-dimensional state-space collapse for a general load balancing system, hence complementing and extending the diffusion approximation results in \cite{kelly1993dynamic} and \cite{teh2002critical}. 

{We would also like to remark that besides load balancing (or scheduling) in parallel servers, state-space collapse result also plays a key role in other heavy-traffic scenarios. For example, given an $n \times n$ input queued switch, it was shown that under the complete resource pooling (CRP) condition (equivalently one-dimensional state-space collapse),  the MaxWeight scheduling algorithm is heavy-traffic delay optimal in the sense of diffusion limit~\cite{stolyar2004maxweight}. When the CRP condition is not met, the state-space would then collapse to a multi-dimensional space instead of a line. In this case, a diffusion limit has been established in~\cite{kang2012diffusion}. For the steady-state behavior in this case, via the drift-based framework, it was showed that the MaxWeight scheduling policy can guarantee optimal delay scaling with respect to $n$~\cite{maguluri2016heavy}~\cite{maguluri2018optimal}.
 % has been recently investigated in scheduling problems~\cite{maguluri2016heavy},~\cite{wangIFIP}. For example, the authors in~\cite{maguluri2016heavy} showed that for an $n \times n$ switch with a multi-dimensional state-space collapse, the MaxWeight scheduling policy can guarantee optimal delay scaling with respect to $n$. 
However, in contrast to the case of the single-dimensional state-space collapse in~\cite{stolyar2004maxweight},~\cite{eryilmaz2012asymptotically}, heavy-traffic delay optimality is not established in~\cite{maguluri2016heavy}~\cite{maguluri2018optimal}. Recently, multi-dimensional state-space collapse was also used to show delay insensitivity of the proportionally fair policy in a bandwidth sharing network in heavy-traffic~\cite{wang2018heavy}.}
We finally remark that the heavy-traffic regime considered in this paper and all the aforementioned papers is the conventional heavy-traffic regime, which is different from the Halfin-Whitt heavy-traffic regime (also known as many-server heavy-traffic regime or quality-and-efficiency-driven regime). In this regime, the heavy-traffic parameter $\epsilon$ approaches zero and the number of servers $N$ goes to infinity at the same time~\cite{armony2005dynamic},~\cite{gurvich2009queue},~\cite{dai2011state}.

\subsection{Notations}
%Given a stochastic process $\{X(t), t\ge 0\}$ with a steady-state distribution, let $\overline{X}$ denote a random variable whose probability distribution is the same as the steady-state distribution of a given process $\{X(t), t\ge 0\}$. 

The dot product in $\mathbb{R}^N$ is denoted by $\inner{\mathbf{x} }{\mathbf{y} } \triangleq \sum_{n=1}^N x_ny_n$. For any $\mathbf{x} \in \mathbb{R}^N$, the $l_1$ norm is denoted by $\norm{\mathbf{x}}_1 \triangleq \sum_{n=1}^N |x_n|$ and $l_2$ norm is denoted by $\norm{\mathbf{x}} \triangleq \sqrt{\inner{\mathbf{x}}{\mathbf{x}}}$. In general, the $l_r$ norm is denoted by $\norm{\mathbf{x}}_r \triangleq (\sum_{n=1}^N |x_n|^r)^{1/r}$. Let $\mathcal{N}$ denote the set $\{1,2,\ldots,N\}$.
% Let $\mathbf{1}_N \triangleq \frac{1}{\sqrt{N}}(1,1,\ldots,1)$. 
% Then the parallel and perpendicular component of any vector $\mathbf{x}$ in $\mathbb{R}^N$ with respect to the vector $\mathbf{1}_N$ is denoted by $\mathbf{x}_\parallel \triangleq \inner{\mathbf{1}_N}{\mathbf{x}}\mathbf{1}_N$ and $\mathbf{x}_\perp \triangleq \mathbf{x} - \mathbf{x}_\parallel$, respectively.

% !TEX root = ./performance18.tex
\section{System Model and Preliminaries}
This section first precisely describes the system model, and then presents several necessary preliminaries.
\subsection{System Model}
We consider a discrete-time load balancing system as follows. There is a central dispatcher and $N$ servers indexed by $n$, each of which maintains a FIFO (first-in, first-out) infinite buffer size queue denoted by $Q_n$. In each time-slot, the central dispatcher routes the new task arrivals to one of the servers as in \cite{eryilmaz2012asymptotically,maguluri2014heavy,wang2016maptask,xie2015priority,xie2016scheduling,zhou2017designing}. Once a task joins a queue, it will remain in that queue until its service is completed.

\subsubsection{Arrival and Service} Let $A_{\Sigma}(t)$ denote the number of exogenous tasks that arrive at the beginning of time-slot $t$. We assume that $A_{\Sigma}(t)$ is an integer-valued random variable, which is \emph{i.i.d.} across time-slots. The mean and variance of $A_{\Sigma}(t)$ are denoted by $\lambda_{\Sigma}$ and $\sigma_{\Sigma}^2$, respectively. We further assume that there is a positive probability for $A_{\Sigma}(t)$ to be zero and the arrival process has a finite support, i.e., $A_{\Sigma}(t) \le A_{\max} < \infty$ for all $t$. 
Let $S_n(t)$ denote the amount of service that server $n$ offers for queue $n$ in time-slot $t$. Note that this is not necessarily equal to the number of tasks that leaves the queue because the queue may be empty. We assume that $S_n(t)$ is an integer-valued random variable, which is \emph{i.i.d.} across time-slots. We also assume that $S_n(t)$ is independent across different servers as well as the arrival process. As before, $S_n(t)$ is also assumed to have a finite support, i.e., $S_n(t) \le S_{\max} <\infty$ for all $t$ and $n$. The mean and variance of $S_n(t)$ are denoted as $\mu_n$ and $\nu_n^2$, respectively. Let $\mu_{\Sigma} \triangleq \Sigma_{n=1}^N \mu_n$ and $\nu_{\Sigma}^2 \triangleq \Sigma_{n=1}^N \nu_n^2$ denote the mean and variance of the hypothetical total service process $S_{\Sigma}(t) \triangleq \sum_{n=1}^N S_n(t)$.

\subsubsection{Queue Dynamics} Let $Q_n(t)$ be the queue length of server $n$ at the beginning of time slot $t$. 
Let $A_n(t)$ denote the number of tasks routed to queue $n$ at the beginning of time-slot $t$ according to the dispatching decision. 
Then the evolution of the length of queue $n$ is given by 
\begin{equation}
	\label{eq:Qdynamic}
	Q_n(t+1) = Q_n(t) + A_n(t) - S_n(t) + U_n(t), n = 1,2,\ldots, N,
\end{equation}
where $U_n(t) = \max\{S_n(t)-Q_n(t)-A_n(t),0\}$ is the unused service due to an empty queue.

\subsection{Preliminaries}
In this paper, we assume that the dispatching decision in each time-slot can at most depend on $\Q(t)$. Thus, with the system model above, the queue length process $\{\Q(t), t\ge 0\}$ forms a Markov chain. We consider a set of load balancing systems $\{\Q^{(\epsilon)}(t), t\ge 0\}$ parameterized by $\epsilon$ such that the mean arrival rate of the exogenous arrival process $\{A_{\Sigma}^{(\epsilon)}(t), t\ge 0\}$ is $\lambda_{\Sigma}^{(\epsilon)} = \mu_\Sigma - \epsilon$. Note that the heavy-traffic parameter $\epsilon$ characterizes the distance between the arrival rate and the capacity region boundary. 

We say that a load balancing system is stable if the Markov chain $\{\Q(t), t\ge 0\}$ is positive recurrent, and then use $\overline{\Q}$ to denote the random vector whose distribution is the same as the steady-state distribution of $\{\Q(t), t\ge 0\}$. Now, we are ready to present the definitions of throughput optimality and steady-state heavy-traffic delay optimality, respectively.

{
\begin{definition}[Throughput Optimal]
	A load balancing policy is said to be throughput optimal if for any arrival rate within the capacity region, i.e., for any $\epsilon > 0$, it can stabilize the system and all the moments of $\big\lVert{\overline{\Q}^{(\epsilon)}}\big\rVert$ are finite.
\end{definition}

Note that this is a stronger definition of throughput optimality than that in~\cite{wang2016maptask}~\cite{xie2015priority}~\cite{zhou2017designing}, because besides the positive recurrence, it also requires all the moments to be finite in steady state for any arrival rate within capacity region.}

% Note that the definition of throughput optimality in this paper is stronger than previous ones. This is because in addition to positive recurrence, it also requires that all the moments are finite in steady-state.

In the heavy-traffic analysis, one is interested in the behavior of the queue lengths as $\epsilon$ approaches zero. In order to present and understand the definition of steady-state heavy-traffic delay optimality, we will first recall the fundamental lower bound on the expected sum queue lengths under any throughput optimal policy \cite{eryilmaz2012asymptotically}.

\begin{lemma}
\label{lem:lower_bound}
    Given any throughput optimal policy and assuming that $(\sigma_{\Sigma}^{(\epsilon)})^2$ converges to a constant $\sigma_{\Sigma}^2$ as $\epsilon$ decreases to zero, then 
	% Let $\overline{\Q}^{(\epsilon)}$ be a random vector which is equal in distribution to the queue length $\Q(t)$ in the steady state under any feasible load balancing scheme. 
	\begin{equation}
	\label{eq:lower_bound}
		\liminf_{\epsilon \downarrow 0} \epsilon \ex{\sum_{n=1}^N \overline{Q}_n^{(\epsilon)} } \ge \frac{\zeta}{2},
	\end{equation}
	where $\zeta \triangleq \sigma_{\Sigma}^2 + \nu_{\Sigma}^2$.
\end{lemma}

The right-hand-side of Eq. \eqref{eq:lower_bound} is the heavy-traffic limit of a hypothetic single-server system with arrival process $A_\Sigma^{(\epsilon)}(t)$ and service process $\sum_n^N S_n(t)$ for all $t\ge0$. This hypothetical single-server queueing system is often called the \textit{resource-pooled system}. Since a task cannot be moved from one queue to another in the load balancing system, it is easy to see that the expected sum queue lengths of the load balancing system is larger than the expected queue length in the resource-pooled system. However, under a certain load balancing policy, the lower bound in Eq. \eqref{eq:lower_bound} can actually be attained in the heavy-traffic limit and hence based on Little's law this policy achieves the minimum average delay of the system in steady-state. This directly motivates the following definition of steady-state heavy-traffic delay optimality as in \cite{eryilmaz2012asymptotically,maguluri2014heavy,wang2016maptask,xie2015priority,xie2016scheduling,zhou2017designing}.

\begin{definition}[Heavy-traffic Delay Optimality in Steady-state]
	A load balancing scheme is said to be heavy-traffic delay optimal in steady-state if the steady-state queue length vector $\overline{\Q}^{(\epsilon)}$ satisfies 
	\begin{equation*}
		\limsup_{\epsilon \downarrow 0} \epsilon \ex{\sum_{n=1}^N \overline{Q}_n^{(\epsilon)} } \le \frac{\zeta}{2},
	\end{equation*}
	where $\zeta$ is defined in Lemma \ref{lem:lower_bound}.
\end{definition}

Before we end this section, we will introduce an $N$-dimensional cone, which will be very useful in our upcoming analysis. In particular, the cone $\mathcal{K}_{\alpha}$ is finitely generated by a set of $N$ vectors $\{\mathbf{b}^{(n)}, n \in \mathcal{N}\}$, i.e., 
\begin{equation}
\label{eq:def_cone}
	\mathcal{K}_{\alpha} = \left\{ \mathbf{x} \in \mathbb{R}^N: \mathbf{x}  = \sum_{n \in \mathcal{N}} w_n\mathbf{b}^{(n)}, w_n \ge 0 \text{ for all } n \in \mathcal{N} \right\},
\end{equation}
where $\mathbf{b}^{(n)}$ is an $N$-dimensional vector with the $n$th component being $1$ and  $\alpha$ everywhere else for some $\alpha \in [0,1]$. It follows that, if $\alpha = 0$, the cone $\mathcal{K}_{\alpha}$ is the non-negative orthant of $\mathbb{R}^N$, and if $\alpha = 1$, the cone $\mathcal{K}_{\alpha}$ reduces to the single-dimensional line in which all the components are equal.
The polar cone $\mathcal{K}_{\alpha}^{\circ}$ of the cone $\mathcal{K}_{\alpha}$ is defined as
\begin{equation*}
	\mathcal{K}_{\alpha}^{\circ} = \left\{ \mathbf{x} \in \mathbb{R}^N: \inner{\mathbf{x}}{\mathbf{y}} \le 0 \text{ for all } \mathbf{y} \in \mathcal{K}_{\alpha} \right\},
\end{equation*}
which will also be quite important in our analysis.

\section{Main Results}
In this section, we present our main results. First, we show that a load balancing policy can be heavy-traffic delay optimal in steady-state even under a multi-dimensional state-space collapse. Then, by leveraging the key insight behind this result, we are able to explore the degree of flexibility a load balancing policy can enjoy while guaranteeing both throughput optimality and heavy-traffic delay optimality. Furthermore, a useful generalization of this result is presented at the end of this section.

\subsection{Multi-dimensional State-space Collapse}
We will first introduce the notion of state-space collapse used in this paper, which intuitively means that in steady-state the queue length process concentrates around a region of the state-space in heavy-traffic. As stated before, in most of the previous works on load balancing, the state-space collapse region is a single-dimensional line. In contrast, we are interested in the situation where the state-space collapse region is the $N$-dimensional cone $\mathcal{K}_{\alpha}$ defined in Eq.~\eqref{eq:def_cone}, which includes the single-dimensional line as a special case.  For a given cone $\mathcal{K}_{\alpha}$, we decompose $\overline{\Q}^{(\epsilon)}$ into two parts as follows
\begin{equation*}
	\overline{\Q}^{(\epsilon)} = \overline{\Q}^{(\epsilon)}_{\parallel} + \overline{\Q}^{(\epsilon)}_{\perp},
\end{equation*}
where $ \overline{\Q}^{(\epsilon)}_{\parallel}$ is the projection onto the cone $\mathcal{K}_{\alpha}$, referred to as the parallel component,  and $\overline{\Q}^{(\epsilon)}_{\perp}$ is the remainder, referred to as the perpendicular component, which is actually the projection onto the polar cone $\mathcal{K}_{\alpha}^{\circ}$. Note that this decomposition is well defined and unique since the cones $\mathcal{K}_{\alpha}$ and $\mathcal{K}_{\alpha}^{\circ}$ are both closed and convex, which follows from the fact that $\mathcal{K}_{\alpha}$ is finitely generated. The norm $\big\lVert{\overline{\Q}^{(\epsilon)}_{\perp}}\big\rVert$ is the distance between $\overline{\Q}^{(\epsilon)}$ and the cone $\mathcal{K}_{\alpha}$. We say that the queue length process concentrates around the cone $\mathcal{K}_{\alpha}$ if the moments of the distance $\big\lVert{\overline{\Q}^{(\epsilon)}_{\perp}}\big\rVert$ are upper bounded by constants. This motivates the following definition.

\begin{definition}[State-space Collapse to $\mathcal{K}_{\alpha}$]
\label{def_collapse}
Given an $\alpha \in (0,1]$, we say the state-space of a load balancing system collapses to the cone $\mathcal{K}_{\alpha}$ if 
\begin{equation}
\label{eq:def_collapse}
	\ex{\norm{{\overline{\Q}^{(\epsilon)}_{\perp}}}^r } \le M_r
\end{equation}
for all $\epsilon \in (0,\epsilon_0)$, $\epsilon_0 > 0$ and for each $r = 1,2,\cdots$, in which $M_r$ are constants that are independent of $\epsilon$.
\end{definition}

\begin{remark}
	It is worth noting that although the result of state-space collapse is often used as the key step in establishing heavy-traffic delay optimality, the upper bound itself holds even when the system is not in the heavy-traffic limit, and this can be of independent interest for analyzing the performance of the system.
\end{remark}

Now, we are ready to present our first main result.
% , which intuitively states that for a given load balancing policy, if the state-space collapses to a cone $\mathcal{K}_{\alpha}$ for some $\alpha \in (0,1]$, then the policy is heavy-traffic delay optimal.

\begin{theorem}
\label{thm:theorem_1}
	Given a throughput optimal load balancing policy, if there exists an $\alpha \in (0,1]$ such that the state-space collapses to the cone $\mathcal{K}_{\alpha}$, then this policy is heavy-traffic delay optimal in steady-state.
\end{theorem}
\begin{proof}
	See Section \ref{sec:proof_theorem_1}.
\end{proof}

\begin{figure}[t]
	\graphicspath{{./Figures/}}
	\centering
	\includegraphics[width=4in]{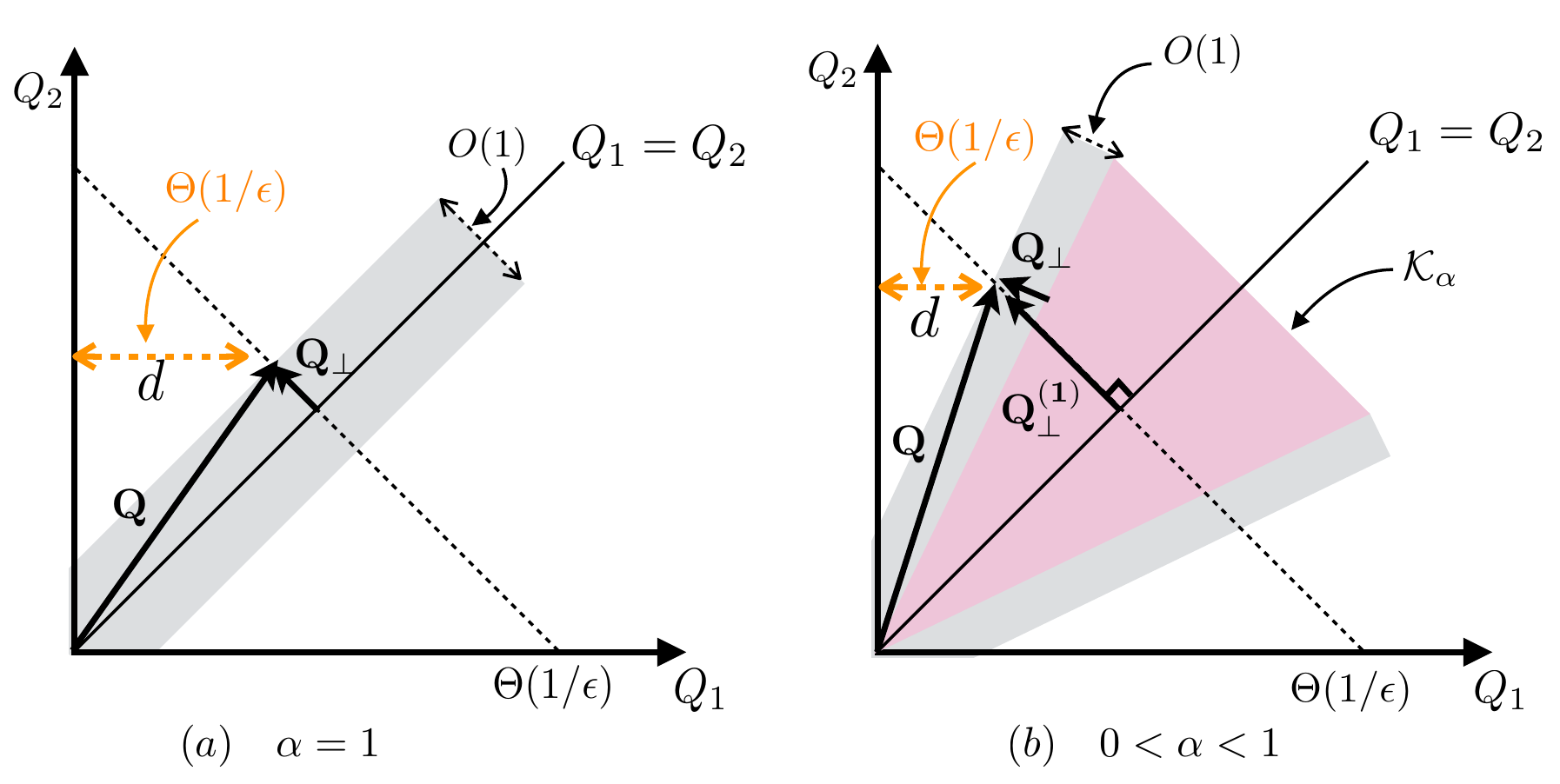}
	\caption{A geometric illustration of the key insight of steady-state heavy-traffic optimality in load balancing systems.\normalfont{ In the figures above, we use the gray area to represent the bounded distance between $\mathbf{Q}$ and the cone $\mathcal{K}_{\alpha}$. The dashed line represents the total tasks in the system and hence has order $1/\epsilon$. In (a), $\alpha = 1$ and hence the system state collapses to the one-dimensional line. As a result, the queue lengths of two servers are nearly equal. However, Theorem \ref{thm:theorem_1} tells us that this is not the key feature of heavy-traffic optimality since in (b) the queue lengths difference between two servers, $\lVert{\mathbf{Q}_{\perp}^{(\mathbf{1})}}\rVert$, is of the same order of $\lVert\mathbf{Q}\rVert$. Rather, the key feature behind heavy-traffic optimality is that it keeps all the servers busy when there is substantial work. This is achieved by keeping the system state far away from the boundary via state-space collapse, as $\epsilon$ approaches zero (see the distance $d$ in both (a) and (b)).}}\label{fig:insight}
	\vspace{-4mm}
\end{figure}

From this theorem, we can make the following important observations regarding heavy-traffic delay optimality in load balancing systems. A geometric illustration is presented in Fig. \ref{fig:insight} to facilitate the understanding.
\begin{enumerate}[(i)]
	\item If $\alpha = 1$, then this theorem reduces to previous results on heavy-traffic delay optimality under a single-dimensional state-space. In this case, the state-space can be regarded as if it evolves in a one-dimensional subspace where all queues are equal. This is because the queue-length difference between servers is bounded by a constant and hence is substantially smaller than the queue lengths themselves, which are on the order of $1/\epsilon$. See Fig.~\ref{fig:insight}(a). 
	% \item Unlike the multi-dimensional state-space collapse in scheduling scenario [][] where the cone lies in a lower subspace, the cone $\mathcal{K}_{\alpha}$ lies in the $N$-dimensional subspace. This highlights the difference between load balancing (or routing) and scheduling.
	\item This theorem tells us that in order to be heavy-traffic delay optimal, a policy does not necessarily have to keep all the queues equal as in the case of $\alpha = 1$. This is because Theorem \ref{thm:theorem_1} implies that heavy-traffic delay optimality is preserved even when the difference between the various queues is of the same order as the queue lengths themselves, as shown in Fig.~\ref{fig:insight}(b). Therefore, this theorem indicates that the key feature of a heavy-traffic delay optimal policy is that \textit{it keeps all the servers busy when there are substantial tasks in the system.} This is achieved by keeping system states far away from the boundary via state-space collapse as $\epsilon$ approaches zero, see the distance $d$ in Fig. \ref{fig:insight}. 
	\item It should also be pointed out that  the cone $\mathcal{K}_{\alpha}$ in Theorem \ref{thm:theorem_1} is not necessarily the actual region that state-space collapses to. In fact, this theorem tells us that for heavy-traffic delay optimality, the actual region of state-space collapse, say $\mathcal{R}$, does not matter as long as it lies within a cone $\mathcal{K}_{\alpha}$ for some $\alpha \in (0,1]$. This is because in this case the distance to the cone $\mathcal{K}_{\alpha}$ is not larger than that to the region $\mathcal{R}$. Thus, once it collapses to the region $\mathcal{R}$, it also collapses to the cone $\mathcal{K}_{\alpha}$ according to Definition~\ref{def_collapse}, and hence achieves heavy-traffic delay optimality. This nice property may be of independent interest since it enables us to establish heavy-traffic delay optimality even when the multi-dimensional state-space collapse region is non-regular and non-convex.

\end{enumerate}

Now, we turn to  provide the high-level intuition on why Theorem~\ref{thm:theorem_1} holds.
To start with, note that the following equation
\begin{equation}
\label{eq:UQ_cross}
	U_n(t)Q_n(t+1) = 0
\end{equation}
holds for all $n$ and $t$. This follows directly from the queue dynamics in Eq. \eqref{eq:Qdynamic}. Thus, for the single-server resource-pooled system, when there is positive unused service at time-slot $t$, the queue must be empty at time-slot $t+1$. In contrast, 
for a load balancing system, due to the fact that a task cannot be moved from one queue to another queue, there exist situations when one queue, say $i$, has positive unused service, i.e., $U_i(t) > 0$ and hence $Q_i(t+1) = 0$, while there are remaining tasks in other queues, i.e., $Q_j(t+1) > 0$ for some $j$. As a result, the average queue length of the resource-pooled system is the lower bound for the load balancing system. Therefore, in order to achieve this lower bound (and hence heavy-traffic delay optimality by definition), a load balancing policy should guarantee that when one queue has positive unused service, all the other queues should be empty in steady-state. In fact, this is actually the insight behind the sufficient and necessary condition for heavy-traffic delay optimality in Lemma \ref{claim_1}, i.e., 
\begin{equation}
	\label{eq:necessary_sketch}
		\lim_{\epsilon \downarrow 0}\ex{\big\lVert\overline{\Q}^{(\epsilon)}(t+1) \big\rVert_1 \big\lVert\overline{\UU}^{(\epsilon)}(t) \big\rVert_1} = 0.
\end{equation}

Next, let us take a closer look at the condition above.  As a result of Eq. \eqref{eq:UQ_cross}, the left-hand side of Eq. \eqref{eq:necessary_sketch} is always zero when all the queue lengths are positive. This explains why it is not necessary to keep all the queue lengths equal as in the case of single-dimensional state-space collapse. Instead, what really matters in the condition is the situation when $U_i(t) > 0$ (and hence $Q_i(t+1) = 0$) for some $i$. In this case, the condition requires all the other queue lengths must be zero as well, i.e., $Q_n(t+1) = 0$ for all $n$. In other words, it requires all the servers be busy when there is substantial work, which is intuitively satisfied when the queue length process collapses to the cone $\mathcal{K}_{\alpha}$ for any $\alpha \in (0,1]$. This is because the cone is far away from the states where some queue is empty while another queue is non-empty in heavy traffic. The proof is presented in Section~\ref{sec:proof_theorem_1}.
\begin{remark}
	We would remark on the choice of the particular form of cone $\mathcal{K}_{\alpha}$, which provides further intuitions on Theorem~\ref{thm:theorem_1}. One reason for the choice is that $\mathcal{K}_{\alpha}$ is finitely generated, and hence it is closed and convex. This guarantees that the projection onto this cone is well-defined and unique. Another more important reason is that $\mathcal{K}_{\alpha}$ can approach the non-negative orthant while guaranteeing that within the cone there are no \textit{`bad points'} where some queue is empty and another queue is non-empty. This directly implies that once the state-space collapses to the cone $\mathcal{K}_{\alpha}$, the condition in Eq. \eqref{eq:necessary_sketch} is satisfied. One might think of using an `ice-cream' cone defined below as a substitute of cone $\mathcal{K}_{\alpha}$,
	\begin{align*}
	\mathcal{K}_{\theta} = \left\{ \mathbf{x} \in \mathbb{R}^N: \frac{ \lVert{\mathbf{x}_{\parallel }^{(\mathbf{1})}}\rVert }{\norm{\mathbf{x}}} \ge \cos(\theta) \right\},
	\end{align*}
	where $\mathbf{x}_{\parallel}^{(\mathbf{1})}$ is the projection of $\mathbf{x}$ onto the line $\mathbf{1} = (1,1,\ldots,1)$. The key problem with such a choice is that in order to exclude all the \textit{`bad points'} from $\mathcal{K}_{\theta}$ for a large system size $N$, the cone $\mathcal{K}_{\theta}$ has to basically reduce to the line $\mathbf{1}$, and hence does not provide us with any further flexibility.
	% The problem with this choice is that \textit{in general} even though the queue length process collapses to a useful cone $\mathcal{K}_{\theta}$, it cannot guarantee heavy-traffic delay optimality. This is because there are possible infinite points within the cone $\mathcal{K}_{\theta}$ such that when one queue is empty and some other queues are non-empty (we call these points \textit{bad points}), and hence violating the necessary condition of heavy-traffic optimality in Eq. \eqref{eq:necessary_sketch}. 
	To see this, let us start with $N=2$. In this case, the choice of $\mathcal{K}_{\theta}$ is fine since it is able to approach the non-negative orthant as $\theta$ approaches $\pi/4$, while guaranteeing that there are no \textit{bad points} within the cone. In fact, it is easy to see that in this case $\theta$ plays the same role as $\alpha$ in $\mathcal{K}_{\alpha}$. Now, consider the case $N = 3$. One might choose $\theta < \arccos(1/\sqrt{3})$ in order to exclude the \textit{`bad points'} on the axes from the cone $\mathcal{K}_{\theta}$. However, all the points on the line $\mathbf{x} = (1,1,0)$ are also \textit{`bad points'}. Thus, in order to exclude all of these points, the choice of $\theta$ should be $[0,\arccos( \sqrt{2} /\sqrt{3}) )$. In general, for the $N$-dimension case, the choice of $\theta$ should be less than $\arccos( \sqrt{N-1} /\sqrt{N})$, which approaches zero for large $N$. Hence, for a large system size $N$, the cone $\mathcal{K}_{\theta}$ has to approach the single-dimensional line $\mathbf{1}$ in order to guarantee heavy-traffic delay optimality, which is not interesting because it does not provide any significant flexibility compared to the single-dimensional line. 
	The insight of keeping all the servers busy by excluding the \textit{`bad points'} when there is substantial work is also useful for scheduling problems. For example,  the cone considered in \cite{maguluri2016heavy} for the scheduling problem in a switch system contains infinitely many \textit{`bad points'}. It is actually due to this problem that the MaxWeight policy in this case can only guarantee optimal scaling rather than heavy-traffic delay optimality under the single-dimensional state-space collapse in \cite{stolyar2004maxweight},~\cite{eryilmaz2012asymptotically}.
	
	In contrast, the $\mathcal{K}_{\alpha}$ considered in our paper is able to exclude all the \textit{`bad points'} for any $\alpha > 0$ and hence guarantees heavy-traffic delay optimality. This fact not only captures the essence of heavy-traffic delay optimality in load balancing systems via multi-dimensional state-space collapse, but also provides flexibility in analyzing and designing new load balancing policies, which will be explored in the next section.
\end{remark}

\subsection{Flexible Load Balancing}
In this section, instead of focusing on yet another policy, we step back and explore the possibility provided by Theorem \ref{thm:theorem_1} in analyzing and more importantly designing flexible load balancing policies that are both throughput optimal and heavy-traffic delay optimal. This is motivated by the fact that existing policies are often too restrictive and might not be easily adopted to guarantee system performance in scenarios where \emph{data locality} or \emph{inaccurate} information of queue lengths exist, which are common in load balancing systems \cite{zaharia2010delay},~\cite{paschos2016routing}.

Before we present our main result, let us first introduce some necessary concepts. Let $P_n(t)$ be the probability that the new arrivals are dispatched to the $n$th shortest queue at time-slot $t$. By the Markovian assumption, the dispatching distribution $\mathbf{P}(t)$ can at most depend on $\Q(t)$.
Let 
\begin{equation}
\label{eq:delta_definition}
	\Delta(t) = \mathbf{P}(t) - \mathbf{P}_{\text{rand}}(t),
\end{equation}
where $\mathbf{P}_{\text{rand}}(t)$ is the dispatching distribution under uniform random routing (homogeneous case) or proportional random routing (heterogeneous case), i.e., for homogeneous servers, each component of $\mathbf{P}_{\text{rand}}(t)$ is $1/N$, and for heterogeneous servers the $n$th component of $\mathbf{P}_{\text{rand}}(t)$ is $\mu_{\sigma_t(n)}/\mu_{\Sigma}$ where $\sigma_t(n)$ is the index of the $n$th shortest queue at time-slot $t$. 

To facilitate the understanding of the concepts above, let us look at some examples. Consider a load balancing system with four homogeneous servers. Under uniform random routing, we have $\Delta(t) = (0,0,0,0)$ for each time-slot $t$. Under the JSQ policy,  the dispatcher always assigns the new arrival to the shortest queue, and thus we have $\Delta(t) = (3/4,-1/4,-1/4,-1/4)$ for each time-slot $t$. Under the power-of-$2$ policy, the dispatcher randomly picks two servers and dispatches the new arrivals to the server with the shorter queue length. It easily follows that $\Delta(t) = (1/4,1/12,-1/12,-1/4)$ for each time-slot $t$. Note that, from these examples, we can see that a positive value of $\Delta_n(t)$ means that the dispatcher favors the $n$th shortest queue, while a non-positive value means the dispatcher disfavors the corresponding queue. This is because $\Delta(t)$ equals $(0,0,0,0)$ under uniform random routing, which has no preference over any queues.

Now, we are prepared to present our second main result, which characterizes the degree of flexibility a load balancing policy enjoys while guaranteeing  throughput and heavy-traffic delay optimality.

\begin{theorem}
\label{thm:theorem_2}
	Given a load balancing policy, if there exists a cone $\mathcal{K}_{\alpha}$ with $\alpha \in (0,1]$ such that for all $\Q(t) \notin \mathcal{K}_{\alpha}$, there is some $k \in \{2,\ldots,N\}$ such that
	\begin{equation}
	\label{eq:condition_1}
		\Delta_n(t) \ge 0, n \le k \text{ and } \Delta_n(t) \le 0, n \ge k
	\end{equation}
	and 
	\begin{equation}
	\label{eq:condition_2}
		\min\left(|\Delta_1(t)|,|\Delta_N(t)|\right) \ge \delta
	\end{equation}
	for some positive constant $\delta$ that is independent of $\epsilon$, then this policy is both throughput and heavy-traffic delay optimal in steady-state.
\end{theorem}
\begin{proof}
	See Section \ref{sec:proof_theorem_2}
\end{proof}
\begin{remark}
	It can be easily seen that previous steady-state heavy-traffic delay optimal policies, namely JSQ and power-of-$d$, satisfy the conditions in Theorem \ref{thm:theorem_2} with $\alpha = 1$.
\end{remark}

Before we turn to the technical aspects, let us first elaborate on the key messages behind this theorem. In sum, this theorem characterizes the flexibility in achieving throughput and heavy-traffic delay optimality from the following two dimensions.
\begin{enumerate}[(i)]
	\item The first dimension relates to the frequency of favoring shorter queues. This can be seen from the fact that there are no requirements on $\Delta(t)$ whenever $\Q(t)$ falls in the cone $\mathcal{K}_{\alpha}$, and $\alpha$ can be arbitrarily close to zero. This is significantly different from previous heavy-traffic delay optimal policies, e.g., JSQ and power-of-$d$. These policies have to favor shorter queues for \textit{every} time-slot and \textit{every} system state. This is  often too restrictive and may not be achievable, especially when considering the data locality problem, since in this case the dispatcher has to place tasks to servers that store the corresponding input data chunks. In contrast, the above theorem tells us that a load balancing policy has the flexibility to adopt any dispatching distribution whenever the queue-length state falls in a region that can be covered by a cone $\mathcal{K}_{\alpha}$ for some $\alpha \in (0,1]$. For example, for a load balancing system with heterogeneous servers, the dispatcher can just use uniform random routing when the system state lies within a cone $\mathcal{K}_{\alpha}$, which provides us a lot of flexibility (e.g., easy implementation and lower message overhead), compared to JSQ policy, and the flexibility increases as $\alpha$ decreases. It is also worth noting that although the delay optimality is preserved in heavy-traffic for any $\alpha \in (0,1]$, the actual delay performance under medium or low loads might get worse as $\alpha$ decreases in some cases. Thus, the parameter $\alpha$ also captures an important trade-off between flexibility and delay performance under medium loads, which may be an interesting open problem to explore in the future.

	\item The second dimension is related to the intensity with which shorter queues are favored. This can be seen from the conditions on $\Delta(t)$ in Eqs. \eqref{eq:condition_1} and \eqref{eq:condition_2}. Specifically, instead of joining only the shortest queue as in the JSQ policy or having monotone decreasing probabilities from joining the shortest queue to the longest queue in the power-of-d policy, an even weaker intensity of favoring shorter queues is sufficient, and this intensity can be characterized by the parameter $\delta$. This kind of flexibility is very useful when the queue length information available at the dispatcher may be inaccurate due to communication delay or sampling error.
	 % Another thing to note is that  even when $\Q(t)$ lies outside $\mathcal{K}_{\alpha}$, the above theorem tells us that the dispatcher still has additional flexibility compared to previous policies in order to achieve heavy-traffic delay optimality. Specifically, it says that $\Delta_n(t)$ does not have to be decreasing with respect to $n$, which is the case in all previous policies. That is, in terms of heavy-traffic optimality, there is no need for the dispatcher to have  the strongest preference on the shortest queue than any other queue.
\end{enumerate}

We now highlight the technical contributions behind this theorem. Since this theorem is proved based on Theorem \ref{thm:theorem_1}, all we need to show are throughput optimality and state-space collapse to the cone.

For throughput optimality, i.e., positive recurrence and bounded moments in steady-state, the standard drift analysis of a suitable Lyapunov function is very difficult in our case. This is because the drift within the cone $\mathcal{K}_{\alpha}$ can be positive; hence, it is challenging from renewal theory to find a sufficiently large $T$ such that the drift within $T$ time slots is negative outside a finite set. On the other hand, the standard fluid approximation method results only in positive recurrence. Our approach is to combine fluid approximations with stochastic Lyapunov theory. In particular, we show that the Lyapunov function used in the fluid model, i.e., the sum of the queue lengths, is also a suitable Lyapunov function for the original stochastic system. This connection allows us to carry out drift analysis in the fluid domain, which makes it easier to find the $T$. Then, we come back to the stochastic system and apply the drift analysis of the same Lyapunov function in the stochastic system to show both positive recurrence and bounded moments in steady-state. This approach also provides us with a good intuition on throughput optimality in load balancing systems. Informally speaking, if  $\Q(t)$ is in the cone $\mathcal{K}_{\alpha}$, the drift of the sum queue lengths is of order $\epsilon$ towards the origin for any dispatching distribution since there is no unused service. If $\Q(t)$ is outside the cone $\mathcal{K}_{\alpha}$, it is easy to see that the dispatching distribution in Theorem \ref{thm:theorem_2} is strictly better than random routing, and hence enjoys a drift towards the origin. Therefore, the sum queue length will not go to infinity in steady-state.

For the state-space collapse to the cone $\mathcal{K}_{\alpha}$, the standard analysis adopted in the single-dimensional collapse also fails in this case. This is because in contrast to the projection onto a line, the projection onto a cone is very difficult. In fact, a closed-form formula of the projection onto a polyhedral cone is still an open problem. To circumvent this difficulty, instead of obtaining the exact projection, we are able to find an important monotone property on the projection, which is sufficient to establish a negative drift independent of $\epsilon$ along the direction of $\Q_{\perp}$ when the queue length state is outside the cone $\mathcal{K}_{\alpha}$. This in turn indicates that the distance between the system state and the cone $\mathcal{K}_{\alpha}$ cannot go to infinity as $\epsilon$ approaches zero. Therefore, by definition, it establishes the result of state-space collapse to the cone. Combining this with throughput optimality yields heavy-traffic delay optimality according to Theorem \ref{thm:theorem_1}. 
\begin{remark}
	We would like to remark that the techniques used to prove Theorem \ref{thm:theorem_2} are of independent interest and may have broader applicability. For example, we are currently investigating whether this technique can be used to design a broader class of heavy-traffic delay optimal \textit{scheduling} policies.
\end{remark}

\subsection{Generalization}

In the last section, we have shown that when it comes to designing a heavy-traffic delay optimal load balancing policy, one has the flexibility of choosing the frequency and intensity of favoring shorter queues, which are parameterized by some fixed positive constants $\alpha$ and $\delta$, respectively. In particular, smaller values of these two constants means favoring the shorter queues less frequently and with less intensity. In this section, we will show that these two constants can actually approach zero at a certain rate with respect to the heavy-traffic parameter $\epsilon$ so that the given policy can still guarantee heavy-traffic delay optimality. As a result, we can exploit this fact to achieve even significant flexibility in designing new policies.
\begin{proposition}
\label{prop:prop_3}
	Given a throughput optimal load balancing policy, if there exists a cone $\mathcal{K}_{\alpha^{(\epsilon)}}$ such that for all $\Q(t) \notin \mathcal{K}_{\alpha^{(\epsilon)}}$, there is some $k \in \{2,\ldots,N\}$ such that
	\begin{equation}
	\label{eq:condition_p1}
		\Delta_n(t) \ge 0, n \le k \text{ and } \Delta_n(t) \le 0, n \ge k
	\end{equation}
	and 
	\begin{equation}
	\label{eq:condition_p2}
		\min\left(|\Delta_1(t)|,|\Delta_N(t)|\right) \ge \delta^{(\epsilon)}
	\end{equation}
	for some $\delta^{(\epsilon)}$. Suppose that $\alpha^{(\epsilon)}$ and $\delta^{(\epsilon)}$ satisfy  
	\begin{align*}
		\alpha^{(\epsilon)}\delta^{(\epsilon)} = \Omega(\epsilon^{\beta})
	\end{align*}
	for some $\beta \in [0,1)$, then this policy is heavy-traffic delay optimal.
\end{proposition}
\begin{proof}
	See Appendix \ref{sec:proof_prop_3}.
\end{proof}
As before, a geometric view of the result of Proposition \ref{prop:prop_3} is presented in Fig. \ref{fig:general}.
\begin{remark}
 It is worth pointing out that under the conditions of Proposition \ref{prop:prop_3}, state-space collapse in this case is different from the one defined in Eq. \eqref{eq:def_collapse}. This is because the moments of the distance between steady state and the cone $\mathcal{K}_{\alpha}$ will go to infinity as $\epsilon$ approaches zero rather than a constant bound in Eq. \eqref{eq:def_collapse}. This type of state-space collapse is sometimes called \textit{multiplicative state-space collapse} as in \cite{kang2009state},~\cite{bramson1998state},~\cite{wangIFIP}. As before, a geometric view of the result of Proposition \ref{prop:prop_3} is presented in Fig. \ref{fig:general}.
\end{remark}

\begin{figure}[t]
	\graphicspath{{./Figures/}}
	\centering
	\includegraphics[width=3.5in]{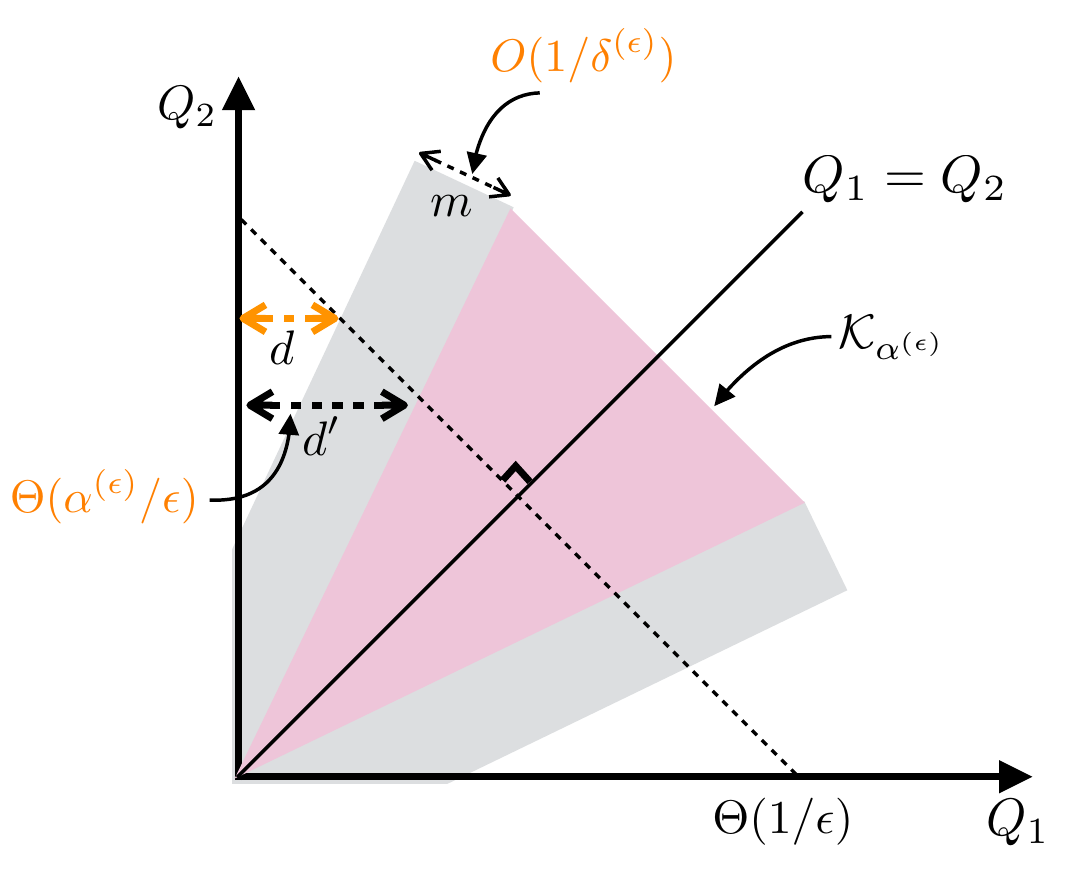}
	\caption{A geometric illustration of the result in Proposition \ref{prop:prop_3}.\normalfont{ As before, we use the gray area to represent the  distance between $\mathbf{Q}$ and the cone $\mathcal{K}_{\alpha}$. The dashed line represents the total tasks in the system and hence has order $1/\epsilon$. As in Fig. \ref{fig:insight}, in order to guarantee that all the servers are busy when there is substantial work in the system, the distance $d$ should be large enough when $\epsilon$ goes to zero. To this end, it is sufficient to require that the order of distance $m$ is smaller than that of the distance $d'$. It is easy to see $\alpha^{(\epsilon)}\delta^{(\epsilon)} = \Omega(\epsilon^\beta)$ for any $\beta \in [0,1)$ satisfies this requirement. }}\label{fig:general}
	\vspace{-4mm}
\end{figure}

% the reason why the queue length of the resource-pooled system serves as the lower bound for the sum queue lengths in the corresponding load balancing system is that 

% Before we end this section, we will discuss why we choose the particular form of the cone as defined in Eq. (). More specifically, we state the reason why a `circular cone' defined as follows does not serve the purpose. 

% \input{math_upper}
% \input{math_lower}
% \input{math_application}
% !TEX root = ./sig2018_winter.tex
\section{Proofs}
In this section, we present the proofs of Theorems \ref{thm:theorem_1} and \ref{thm:theorem_2}, respectively. 

\subsection{Proof of Theorem \ref{thm:theorem_1}}
\label{sec:proof_theorem_1}
Before we present the proof, we first introduce the following lemma.
\begin{lemma}
    \label{lem:unused_service}
    For any $\epsilon >0$ and $t\ge 0$, we have 
    \begin{equation*}
    \label{eq:QtimesU}
    Q_n^{(\epsilon)}(t+1)U_n^{(\epsilon)}(t) = 0.
    \end{equation*}
    Moreover, if the system has a finite first moment, then we have for some constants $c_1$ and $c_r$
 	\begin{equation*}
 	\label{eq:ule}
 		\ex{\big\lVert\overline{\UU}^{(\epsilon)}\big\rVert_1^2} \le c_1 \epsilon \text{ and } \ex{\big\lVert\overline{\UU}^{(\epsilon)}\big\rVert^r_r} \le c_r \epsilon,
 	\end{equation*}
 	where $r \in (1,\infty)$.
\end{lemma}
\begin{proof}
	According to the queues dynamic in Eq. \eqref{eq:Qdynamic}, we can see that when $U_n(t)$ is positive, $Q_n(t+1)$ must be zero, which directly implies the result $ Q_n^{(\epsilon)}(t+1)U_n^{(\epsilon)}(t) = 0$ for any $\epsilon >0$, $1\le n \le N$ and $t\ge 0$. To show the second result, let us consider the Lyapunov function $W_1(\Q(t)) \triangleq \norm{\Q(t)}_1$. Since the system has a finite first moment, the mean drift of $W_1(\Q)$ is zero in steady state, which gives 
	\begin{equation*}
		\ex{\big\lVert\overline{\UU}^{(\epsilon)}\big\rVert_1} = \epsilon.
	\end{equation*}
	Then, due to the fact that $U_n(t) \le S_{\max}$ for all $ 1 \le n \le N$ and $t \ge 0$, we have $\big\lVert \overline{\UU}^{(\epsilon)} \big\rVert^r_r \le (S_{\max})^{r-1}\big\lVert\overline{\UU}^{(\epsilon)}\big\rVert_1$, which implies that $c_r = (S_{\max})^{r-1}$. Note that $\big\lVert\overline{\UU}^{(\epsilon)}\big\rVert_1^2 \le N \big\lVert\overline{\UU}^{(\epsilon)}\big\rVert_2^2 $, which gives $c_1 = NS_{\max}$.
\end{proof}

Now we are ready to prove Theorem \ref{thm:theorem_1}.

\begin{proof}[Proof of Theorem \ref{thm:theorem_1}]
	We prove this theorem by combining the following lemma with the condition of state-space collapse to the cone $\mathcal{K}_{\alpha}$. The proof of the lemma is relegated to Appendix \ref{sec:proof_of_claim_1}.
	
\begin{lemma}
\label{claim_1}
	For a throughput optimal policy, it is heavy-traffic delay optimal if and only if 
	\begin{equation}
	\label{eq:necessary}
		\lim_{\epsilon \downarrow 0}\ex{\big\lVert\overline{\Q}^{(\epsilon)}(t+1) \big\rVert_1 \big\lVert\overline{\UU}^{(\epsilon)}(t) \big\rVert_1} = 0.
	\end{equation}
\end{lemma}

	Next, we will show that under the condition that the state space collapses to a cone $\mathcal{K}_{\alpha}$ with $\alpha \in (0,1]$, the condition in Eq. \eqref{eq:necessary} holds. For brevity, we will omit the references $t$ and $\epsilon$, and use $\overline{\Q}^+$ to denote $\overline{\Q}(t+1)$ in the following. First, we have 
	\begin{align}
	\label{eq:cross_term}
		\mathcal{T}^{(\epsilon)} &\triangleq \ex{\big\lVert\overline{\Q}^{(\epsilon)}(t+1) \big\rVert_1 \big\lVert\overline{\UU}^{(\epsilon)}(t) \big\rVert_1}\nonumber \\
		& = \ex{\sum_{i=1}^N \overline{U}_i \left(\sum_{j=1}^N \overline{Q}_j^+ \right) }\nonumber\\
		& = \ex{\sum_{i=1}^N \overline{U}_i \left(\sum_{j=1}^N \left(\overline{Q}_{\parallel j}^+  +  \overline{Q}_{\perp j}^+ \right) \right) },
	\end{align}
	where $\overline{Q}_{\parallel j}^+$ is the $j$-th component of $(\overline{\Q}^+)_\parallel$ and similarly $\overline{Q}_{\perp j}^+$ is the $j$-th component of $(\overline{\Q}^+)_\perp$. For simplicity, we use $\overline{\Q}_\parallel^+$ to denote $(\overline{\Q}^+)_\parallel$ and $\overline{\Q}_\perp^+$ to denote $(\overline{\Q}^+)_\perp$, respectively. Since the vector $\overline{\Q}_\parallel^+$ is in cone $\mathcal{K}_{\alpha}$ by definition, there exist non-negative weights $w_1,\ldots,w_N$ such that $\overline{\Q}_\parallel^+ = \sum w_n \mathbf{b}^{(n)}$. Recall that when $U_n(t) > 0$, $Q_n(t+1) = 0$ by Lemma \ref{lem:unused_service}. Thus, when $\overline{U}_i(t) > 0$, we have 
	\begin{align}
		\overline{Q}_i^+ &= 0 \nonumber\\
		\overline{Q}_{\parallel i}^+  &= - \overline{Q}_{\perp i}^+\nonumber\\
		\sum w_n b^{(n)}_i & = \overline{Q}_{\parallel i}^+ \nonumber\\
		\sum w_n b^{(n)}_j &  = \overline{Q}_{\parallel j}^+  \le \frac{1}{\alpha}\overline{Q}_{\parallel i}^+ \text{ for all } j \neq i \nonumber
	\end{align}
	The last inequality follows from the definition of vector $\mathbf{b}^{(n)}$. Therefore, the term $\mathcal{T}$ in Eq. \eqref{eq:cross_term} can be upper bounded as follows.
	\begin{align}
	\label{eq:upper_T}
		\mathcal{T}^{(\epsilon)} &\lep{a} \ex{\sum_i \overline{U}_i \left( -N_1 \overline{Q}_{\perp i}^+  + \sum_j \overline{Q}_{\perp j}^+ \right)}\nonumber\\
		& = \ex{\inner{\overline{\UU}}{-N_1\overline{\Q}^+_{\perp}}} + \ex{\inner{\overline{\UU}} {   \inner{\mathbf{1}}{\overline{\Q}^+_{\perp}} \mathbf{1}}}\nonumber\\
		& \lep{b} \ex{\inner{\overline{\UU}}{-N_1\overline{\Q}^+_{\perp}}}\nonumber\\
		& \lep{c} N_1 \sqrt{\ex{\norm{\overline{\UU}}^2 }  \ex{\norm{\overline{\Q}^+_{\perp}}^2}}.\nonumber\\
		& \lep{d} N_1 \sqrt{c_2\epsilon M_2}
	\end{align}
	where in (a) $N_1 = N/\alpha$; (b) comes from the non-negativity of $\overline{\UU}$ and  the fact that $\inner{\mathbf{1}}{\overline{\Q}^+_{\perp}} \le 0$ since $\mathbf{1} \in \mathcal{K}_{\alpha}$ and ${\overline{\Q}^+_{\perp}} \in \mathcal{K}_{\alpha}^{\circ}$; (c) is the result of Cauchy-Schwartz inequality for random vectors; (d) holds because of Lemma \ref{lem:unused_service} and the definition of state-space collapse in Eq. \eqref{eq:def_collapse} combined with the fact that $\Q(t+1)$ and $\Q(t)$ have the same distribution in steady-state. Since $c_2$, $M_2$ and $N_1$ are all constants that are independent of $\epsilon$, we have $\lim_{\epsilon \to 0} \mathcal{T}^{(\epsilon)} = 0$, which establishes the result in Eq. \eqref{eq:necessary}, and hence heavy-traffic delay optimality.
	\end{proof}

\subsection{Proof of Theorem \ref{thm:theorem_2}}
\label{sec:proof_theorem_2}
As already pointed out, the proof of Theorem \ref{thm:theorem_2} naturally falls into two parts: throughput optimality and state-space collapse. 
Then, it follows directly from Theorem \ref{thm:theorem_1} that the result in Theorem \ref{thm:theorem_2} is true. In both proofs, we will use the Lyapunov drift-based approach developed in \cite{eryilmaz2012asymptotically} to derive bounded moments in steady state. The following lemma is a $T$-step version of Lemmas 2 and 3 in \cite{maguluri2016heavy}. This lemma could be proved by simply replacing the one-step transition probability to $T$-step transition probability, and hence we omit the proof here.
\begin{lemma}
	      \label{lem:basis}
	        For an irreducible aperiodic and positive recurrent Markov chain $\{X(t), t \ge 0\}$ over a countable state space $\mathcal{X}$, which converges in distribution to $\overline{X}$,  and suppose $V: \mathcal{X} \rightarrow \mathbb{R}_{+}$ is a Lyapunov function. We define the $T$ time-slot drift of $V$ at $X$ as 
	        \[\Delta V(X)\triangleq [V(X(t_0+T)) - V(X(t_0))] \mathcal{I}(X(t_0) = X),\]
	        where $\mathcal{I}(.)$ is the indicator function. Suppose for some positive finite integer $T$, the $T$ time-slot drift of $V$ satisfies the following conditions:

	        \begin{itemize}
	          \item (C1) There exists an $\eta> 0$ and a $\kappa <  \infty$ such that for any $t_0 = 1,2,\ldots$ and for all $X \in \mathcal{X}$ with $V(X)\ge \kappa$, 
	          \[\mathbb{E}\left[\Delta V(X) \mid X(t_0) = X\right]\le -\eta.\]
	          \item (C2) There exists a constant $D < \infty$ such that for all $X\in \mathcal{X}$,
	          \[\mathbb{P}(|\Delta V(X)| \le D) = 1.\]
	        \end{itemize}

	        Then $\{V(X(t)), t\ge0\}$ converges in distribution to a random variable $\overline{V}$, and all moments of $\overline{V}$ exist and are finite. More specifically, we have for any $r = 1,2,\ldots$
	        \begin{equation}
	        \label{eq:upper_siva}
	        	\ex{V(\overline{X})^r} \le (2\kappa)^r + (4D)^r\left(\frac{D+\eta}{\eta} \right)^r r!.
	        \end{equation}
\end{lemma}

% !TEX root = ./sig2018_winter.tex

\subsubsection{Throughput Optimality} We would  prove the following result in this subsection. 
\begin{proposition}
\label{prop:throughput}
	Under the condition of Theorem \ref{thm:theorem_2}, the given policy is throughput optimal. 
\end{proposition}

We would prove this result by combining fluid approximations with stochastic Lyapunov theory. Thus, let us first introduce some necessary notations and useful lemmas. In order to distinguish from stochastic analysis, we define $\mathcal{X} \triangleq (\mathcal{X}(t),t=0,1,2,\ldots)$, in which $\mathcal{X}(t) \triangleq (Q_1(t), Q_2(t), \ldots, Q_n(t) )$ in fluid domain. Then under our assumption and the queue-length based policy, $\mathcal{X}  = (\mathcal{X}(t),t=0,1,2,\ldots)$ is a discrete-time countable Markov chain. That is, the system state is denoted by $\mathcal{X}$ in the fluid approximation analysis.
To establish the fluid model of $\mathcal{X}$, we need several notations. Let us define the norm of $\mathcal{X}(t)$ as $\norm{\mathcal{X}(t)}_1 \triangleq \sum_{n=1}^N Q_n(t)$. Let $\mathcal{X}^{(x)}$ denote a process $\mathcal{X}$ with an initial state satisfying 
\begin{equation}
\label{eq:norm}
	\norm{\mathcal{X}^{(x)}(0)}_1 = x.
\end{equation}
Let $\mathcal{A}_i(t)$ and $\mathcal{D}_i(t)$ denote the \textit{accumulated} arrival and actual departure tasks at queue $i$ up to time-slot $t$, respectively. $\mathcal{A}_{\Sigma}(\tau)$ denotes the \textit{accumulated} exogenous arrivals for a given $\tau$ units of time-slots. $\mathcal{S}_i(\tau)$ denotes the \textit{accumulated} offered service for queue $i$ during a given $\tau$ units of time-slots. Moreover, let $\mathcal{G}_i(t)$ denote the accumulated number of time-slots up to time-slot $t$ in which the new arrivals are routed to queue $i$, and let $\mathcal{B}_i(t) \triangleq \sum_0^t 1\{Q_i(s)>0\}$ denote the accumulated number of time-slots up to time-slot $t$ in which queue $i$ is busy. We also adopt the convention that $\mathcal{A}_i(0) = 0$, $\mathcal{D}_i(0) = 0$, $\mathcal{G}_i(0) = 0$  and $\mathcal{B}_i(0) = 0$.  Therefore, we have $\mathcal{A}_i(t) = \mathcal{A}_{\Sigma}(G_i(t)) \le \mathcal{A}_{\Sigma}(t)$ and $\mathcal{D}_i(t) = \mathcal{S}_i(\mathcal{B}_i(t)) \le \mathcal{S}_i(t)$. Then the queue length $Q_i$ can be described in an alternative form as follows
\begin{equation}
\label{eq:q_dynamic}
	Q_i(t) = Q_i(0) + \mathcal{A}_i(t) - \mathcal{D}_i(t).
\end{equation}

Let us define another process $\mathcal{Y} \triangleq (Q,\mathcal{A},\mathcal{D},\mathcal{A}_{\Sigma},\mathcal{S},\mathcal{G},\mathcal{B})$, i.e., a tuple that denotes a list of processes, and clearly, a sample path of $\mathcal{Y}^{(x)}$ uniquely determines the sample path of $\mathcal{X}^{(x)}$. Then, we extend the definition of $\mathcal{Y}$ to each continuous time $t \ge 0$ as $\mathcal{Y}^{(x)}(t) \triangleq \mathcal{Y}^{(x)}(\lfloor t \rfloor)$.
Recall that a sequence of functions $f_n(\cdot)$ is said to converge to a function $f(\cdot)$ uniformly over compact (u.o.c) interval if for all $t \ge 0$, $\lim_{n \to \infty}\sup_{0 \le t^{\prime} \le t} \left|f_n(t^{\prime}) - f(t^{\prime})\right| = 0$. We now consider a sequence of process $\left\{ \frac{1}{x_n}\mathcal{Y}^{(x_n)}(x_n\cdot) \right\}$, which is scaled both in time and space, and show the convergence properties of the sequence in the following lemma.
\begin{lemma}
\label{lemma:fluid model}
	With probability one, for any sequence of the process $\{\frac{1}{x_n}\mathcal{Y}^{(x_n)}(x_n \cdot)\}$, where $x_n$ is a sequence of positive integers with $x_n \to \infty$, there exists a subsequence $x_{n_k}$ with $x_{n_k} \to \infty$ as $k \to \infty$ such that the following u.o.c convergences hold:
	\begin{equation}
	\label{eq:qq}
		 \frac{1}{x_{n_k}}{Q}_i^{(x_{n_k})}(x_{n_k} t) \to q_i(t)
	\end{equation}
	\begin{equation}
	\label{eq:ai}
		 \frac{1}{x_{n_k}}{\mathcal{A} }_i^{(x_{n_k})}(x_{n_k} t) \to a_i(t)
	\end{equation}
	\begin{equation}
	\label{eq:di}
		 \frac{1}{x_{n_k}}{  \mathcal{D}  }_i^{(x_{n_k})}(x_{n_k} t) \to d_i(t)
	\end{equation}
	\begin{equation}
	\label{eq:sigma}
		 \frac{1}{x_{n_k}}{ \mathcal{A}  }_{\Sigma}^{(x_{n_k})}(x_{n_k} t) \to a_{\Sigma}(t)
	\end{equation}
	\begin{equation}
	\label{eq:s}
		 \frac{1}{x_{n_k}}{  \mathcal{S}   }_{i}^{(x_{n_k})}(x_{n_k} t) \to s_{i}(t)
	\end{equation}
	\begin{equation}
	\label{eq:g}
		 \frac{1}{x_{n_k}}{  \mathcal{G}  }_{i}^{(x_{n_k})}(x_{n_k} t) \to g_{i}(t)
	\end{equation}
	\begin{equation}
	\label{eq:b}
		 \frac{1}{x_{n_k}}{ \mathcal{B}  }_{i}^{(x_{n_k})}(x_{n_k} t) \to b_{i}(t)
	\end{equation}
	where $q_i$, $a_i$, $d_i$, $a_{\Sigma}$, $s_i$, $g_i$ and $b_i$ are some Lipschitz continuous functions in $[0,\infty)$. Hence all the functions are differentiable at almost every time $ t \in [0,\infty)$, which is called \textit{regular time}.
\end{lemma}
\begin{proof}
	See Appendix \ref{appendex:proof_fluid_model}.
\end{proof}

The fluid model of our considered load balancing system is given by the following lemma. Note that the fluid model holds for any work-conserving FIFO and queue-length based policy.

\begin{lemma}
\label{lemma:fluid limit}
	Any fluid limit $(q_i,a_i,d_i,a_{\Sigma},s_i,g_i,b_i)$ satisfies the following equations
	\begin{equation}
	\label{eq:fluid_q}
		q_i(t) = q_i(0) + a_i(t) - d_i(t)
	\end{equation}
	\begin{equation}
	\label{eq:fluid_ai}
		a_i(t) = \lambda g_i(t)
	\end{equation}
	\begin{equation}
	\label{eq:fluid_di}
		d_i(t) = \mu_i b_i(t)
	\end{equation}
	\begin{equation}
	\label{eq:fluid_a_all}
		a_{\Sigma}(t) = \lambda t
	\end{equation}
	\begin{equation}
	\label{eq:fluid_s}
		s_{i}(t) =  \mu_i t
	\end{equation}
	\begin{equation}
	\label{eq:sum_g}
		\sum_i^{n}g_i(t) = t
	\end{equation}
	and for any regular time $t$, we have 
	\begin{equation}
	\label{eq:derivative}
		q_i^{\prime}(t) =
		\begin{cases} 
      	\lambda g_i^{\prime}(t) - \mu_i, & q_i(t) > 0. \\
      	0, & q_i(t) = 0. \\
   		\end{cases}
	\end{equation}	
\end{lemma}
\begin{proof}
	See Appendix \ref{appendex:proof_fluid_limit}.
\end{proof}

Now we are well prepared to present the proof of Proposition~\ref{prop:throughput} on throughput optimality.

\begin{proof}[Proof of Proposition \ref{prop:throughput}]
First, recall the permutation $\sigma_t(\cdot)$ of $(1,2,\ldots,N)$ which satisfies $Q_{\sigma_t(1)}(t) \le Q_{\sigma_t(2)}(t) \le \ldots Q_{\sigma_t(N)}(t)$ and ties are broken randomly. Now we can establish the following claim, the proof of which is relegated to Appendix \ref{sec:proof_of_Claim_2}.

\begin{claim}
\label{claim_2}
If $ q_{\sigma_t(1)}(t) = q_{\sigma_t(2)}(t) = q_{\sigma_t(m)}(t) = 0 <  q_{\sigma_t(m+1)}(t) \le \ldots \le q_{\sigma_t(N)}(t)$ for some $1\le m <N$, then 
\begin{equation*}
	\sum_{n=m+1}^N g^\prime_{\sigma_t(n)} = \sum_{n=m+1}^N \left(\Delta_n(t) + \frac{\mu_{\sigma_t(n)}}{\mu_{\Sigma}}\right),
\end{equation*}
and $\Delta(t)$ satisfies the conditions in Eq. \eqref{eq:condition_1} and Eq. \eqref{eq:condition_2} of Theorem~\ref{thm:theorem_2}.
\end{claim}
Now, let us consider a Lyapunov function $V(\mathbf{z}) \triangleq \norm{z}_1$. We would like to show that $\dot{V}(\mathbf{q}(t)) \le -l$ for some constant $l > 0$ whenever $V(\mathbf{q}(t)) > 0$. There are two cases to consider. 

\textbf{Case 1:} $q_i(t) > 0$ for all $i \in \{0,1,\ldots,N\}$. 

	In this case, by Eqs. \eqref{eq:derivative} and \eqref{eq:sum_g}, we have 
	\begin{equation}
		\dot{V}(\mathbf{q}(t)) = \sum_{n=1}^{N}q_n^{\prime}(t) = \lambda(\sum_{n=1}^{N}g_n^{\prime}(t)) - \sum_{n=1}^{N}\mu_i = \lambda - \sum_{n=1}^{N}\mu_n = -\epsilon.
	\end{equation}

\textbf{Case 2:} For some $1\le m <N$, $ q_{\sigma_t(1)}(t) = q_{\sigma_t(2)}(t) = q_{\sigma_t(m)}(t) = 0 <  q_{\sigma_t(m+1)}(t) \le \ldots \le q_{\sigma_t(N)}(t)$.

	In this case, we have 
	\begin{align}
		\dot{V}(\mathbf{q}(t)) & \ep{a} \sum_{n=m+1}^{N}q_n^{\prime}(t)\nonumber  \\
		&\ep{b} \sum_{n=m+1}^N \lambda\left(\Delta_n(t) + \frac{\mu_{\sigma_t(n)}}{\mu_{\Sigma}}\right) - \sum_{n=m+1}^N \mu_{\sigma_t(n)}\nonumber\\
		&\lep{c} \sum_{n=m+1}^N \lambda \frac{\mu_{\sigma_t(n)}}{\mu_{\Sigma}}- \sum_{n=m+1}^N \mu_{\sigma_t(n)}\nonumber\\
		&\lep{d} -\epsilon \frac{\mu_{\min}}{\mu_{\Sigma}}\nonumber
	\end{align}
	where (a) comes from Eq. \eqref{eq:derivative}; (b) follows from Claim \ref{claim_2}; (c) holds due to the fact that $\sum_{n=m+1}^N \Delta_n(t) \le 0$ when $\Delta(t)$ satisfies Eqs. \eqref{eq:condition_1} and \eqref{eq:condition_2}; (d) is true since $\lambda = \mu_{\Sigma} - \epsilon$ and $\mu_{\min} = \min_{ 1 \le n \le N }(\mu_n)$.

	Therefore, combining the above two cases, yields
	\begin{equation*}
		\dot{V}(\mathbf{q}(t)) \le -l \text{ whenever } V(\mathbf{q}(t)) > 0
	\end{equation*}
	where $l \triangleq -\epsilon \mu_{\min}/\mu_{\Sigma} > 0$. This result implies that for any $\gamma \in (0,1)$, there exists a finite $T$ such that $V(\mathbf{q}(T)) \le \gamma$. Now, consider any fixed sequence of processes $\{\mathcal{X}^{(x)}, x = 1,2,\ldots\}$ (for simplicity also denoted as $\{x\}$). Then, from the convergence in Lemma \ref{lemma:fluid model}, we have that for any subsequence $\{x_n\}$ of $\{x\}$, there exists a further (sub)subsequence $\{x_{n_k}\}$ with probability one such that 
	\begin{equation*}
	\lim_{k \to \infty} \frac{1}{x_{n_k}}\norm{\mathcal{X}^{(x_{n_k})}(x_{n_k}T)}_1  = \sum_i^{n}|q_i(T)| \le \gamma \triangleq 1-\xi.
	\end{equation*}
	This further implies that with probability one, 
	\begin{equation*}
		\limsup_{x \to \infty}\left[\frac{1}{x}\norm{\mathcal{X}^{(x)}(xT)}_1\right] \le 1-\xi
	\end{equation*}
	holds, because there is always a subsequence of $\{x\}$ that converges to the same limit as $\limsup_{x \to \infty}\left[\frac{1}{x}\norm{\mathcal{X}^{(x)}(xT)}_1\right].$

	According to Eq. \eqref{eq:q_dynamic}, we have $\norm{\mathcal{X}^{(x)}(xT)}_1 \le x + \sum_{i=1}^{N}\mathcal{A}_i(xT) $. Hence,
	\begin{equation*}
		\mathbb{E}\left[\frac{1}{x}\norm{\mathcal{X}^{(x)}(xT)}_1 \right] \le 1 + \lambda T \le \infty.
	\end{equation*}
	Therefore, from the dominated convergence theorem, we have 
	\begin{equation*}
		\limsup_{x \to \infty}\mathbb{E}\left[\frac{1}{x}\norm{\mathcal{X}^{(x)}(xT)}_1\right] = \mathbb{E}\left[\limsup_{x \to \infty}\frac{1}{x}\norm{\mathcal{X}^{(x)}(xT)}_1\right] \le 1-\xi.
	\end{equation*} 
	This result in turn implies that there exists an $x_0$ such that for all $x = \norm{\mathcal{X}^{(x)}(0)}_1 \ge x_0$ 
	\begin{equation}
	\label{eq:drift_xT}
		\ex{\norm{\mathcal{X}^{(x)}(xT)}_1 - \norm{\mathcal{X}^{(x)}(0)}_1} \le -\frac{\xi x_0}{2}.
	\end{equation}

	Now, let us turn to the stochastic analysis of the Lyapunov drift. In particular, we consider the mean drift of Lyapunov function $V(\Q(t)) = \norm{\Q(t)}_1$. We need to show that the Lyapunov function $V(.)$ satisfies the conditions (C1) and (C2) in Lemma \ref{lem:basis}, respectively.

	For Condition (C2), we have 
    \begin{equation*}
        \begin{split}
          |\Delta V(\Q)| &= \big| \norm{\Q(t_0+T)}_1 - \norm{\Q(t_0)}_1 \big| \mathcal{I}(\Q(t_0) = \Q)\\
          & \lep{a} \norm{\Q(t_0+T) - \Q(t_0)}_1\mathcal{I}(\Q(t_0) = \Q)\\
          & \lep{b} T N\max(A_{\max},S_{\max})
        \end{split} 
    \end{equation*}
    where (a) follows from the fact that  $|\norm{{\bf x}}_1 - \norm{{\bf y}}_1| \le \norm{{\bf x} - {\bf y}}_1$ holds for any ${\bf x}$, ${\bf y} \in \mathbb{R}^N$;  (b) holds due to the assumptions that the $A_\Sigma(t) \le A_{\max}$ and $S_n(t) \le S_{\max}$ for all $t \ge 0$ and all $1\le n\le N$, and are independent of the queue length. This establishes the condition (C2) in Lemma \ref{lem:basis}.

    For Condition (C1), we have
    \begin{align}
    	&\ex{\Delta V({\Q})\mid \Q(t_0)=\Q}\nonumber\\
          = &\ex{\norm{\Q(t_0+T_1)}_1 - \norm{\Q(t_0)}_1 \mid \Q(t_0) = \Q}\nonumber\\
          \ep{a} &\ex{\norm{\Q(T_1)}_1 - \norm{\Q(0)}_1 \mid \Q(0) = \Q}\nonumber\\
          \lep{b}& -\frac{\xi x_0}{2}\nonumber
    \end{align}
    where (a) follows from the i.i.d assumption of exogenous arrival and service, and the system is Markovian with respect to the vector of queue lengths; (b) holds for $T_1 = x_0 T$ and $V(\Q(0)) \ge x_0$. This directly comes from Eq. \eqref{eq:drift_xT} and the fact $\norm{\mathcal{X}(t)}_1 = \sum_{n=1}^N Q_n(t)$. Hence, it establishes the condition (C1) in Lemma \ref{lem:basis}, and thus throughput optimality.
\end{proof}

% !TEX root = ./performance18.tex
\subsubsection{State-space Collapse to Cone} We would prove the following result in this subsection, which combined with the throughput optimality in the last subsection directly implies heavy-traffic delay optimality according to Theorem \ref{thm:theorem_1}.
\begin{proposition}
\label{prop:collapse}
	Under the condition of Theorem \ref{thm:theorem_2}, the state-space in steady-state collapses to the cone $\mathcal{K}_\alpha$, {i.e., there exists $\epsilon_0 ={\mu_{\Sigma}\delta}/{(4N+2\delta)}$ }such that for all $\epsilon \in (0,\epsilon_0)$
	\begin{equation}
	\ex{\norm{{\overline{\Q}^{(\epsilon)}_{\perp}}}^r } \le M_r
\end{equation}
holds for each $r = 1,2,\cdots$, in which $M_r$ are constants that are independent of $\epsilon$.
\end{proposition}

Before we prove Proposition \ref{prop:collapse}, we first define the following Lyapunov functions and their corresponding drifts.
      \begin{equation*}
        V_\perp(\Q) \triangleq \norm{\Qc}, W(\Q) \triangleq \norm{\Q}^2 \text{ and } W_\parallel(\Q)\triangleq \norm{\Qp}^2
      \end{equation*}
      with the corresponding one time-slot drift given by
      \begin{equation*}
        \begin{split}
            &\Delta V_\perp(\Q)\triangleq [V_\perp(\Q(t_0+1)) - V_\perp(\Q(t_0))] \mathcal{I}(\Q(t_0) = \Q)\\
            &\Delta W(\Q)\triangleq [W(\Q(t_0+1)) - W(\Q(t_0))] \mathcal{I}(\Q(t_0) = \Q)\\
            &\Delta W_\parallel(\Q)\triangleq [W_\parallel(\Q(t_0+1)) - W_\parallel(\Q(t_0))] \mathcal{I}(\Q(t_0) = \Q)\\
        \end{split}   
      \end{equation*}

Now, we are ready to prove Proposition \ref{prop:collapse}.
\begin{proof}[Proof of Proposition \ref{prop:collapse}]
	To establish the bounded moments of $\norm{\Qc}$, based on Lemma \ref{lem:basis}, all we need to show is that the the drift of Lyapunov function $V_\perp(.)$ satisfies the two conditions for all $\epsilon \in (0,\epsilon_0)$.
    For condition (C2), we have 
    \begin{align}
    \label{eq:boundedQc}
            & |\Delta V_\perp(\Q)| \nonumber\\
            = & | \norm{\Qc(t_0+1)} - \norm{\Qc(t_0)} | \mathcal{I}(\Q(t_0) = \Q)\nonumber \\
            \lep{a} & \norm{\Qc(t_0+1) - \Qc(t_0)}\mathcal{I}(\Q(t_0) = \Q)\nonumber\\
            \lep{b} & \norm{\Q(t_0+1) - \Q(t_0)} \mathcal{I}(\Q(t_0) = \Q)\nonumber\\
            \lep{c} & \sqrt{N} \max(A_{\max},S_{\max})  
    \end{align}
    where  (a) follows from the fact that  $|\norm{{\bf x}} - \norm{{\bf y}}| \le \norm{{\bf x} - {\bf y}}$ holds for any ${\bf x}$, ${\bf y} \in \mathbb{R}^N$; (b) follows from the non-expansive property of projection and the fact that $\Qc$ is the projection onto the convex closed cone $\mathcal{K}_{\alpha}^{\circ}$. (c) holds due to the assumptions that the $A_\Sigma(t) \le A_{\max}$ and $S_n(t) \le S_{\max}$ for all $t \ge 0$ and all $ 1 \le n \le N$, and are both independent of queue lengths. This verifies Condition (C2) in Lemma \ref{lem:basis}.

    For condition (C1), we need the following result, the proof of which is relegated to Appendix \ref{sec:proof_of_Claim_3}.
    \begin{claim}
    \label{claim_3}
    For any $t\ge0$, we have 
    	\begin{equation*}
		\begin{split}
			&\ex{\Delta V_\perp(\Q) \mid \Q(t) = \Q }\\
		   \le & \frac{1}{2\norm{\Qc(t)}} \ex{\left(2\inner{\Qc(t)}{\A(t) - \s(t)} + L\right) \mid \Q(t) = \Q}
		\end{split}
	\end{equation*}
	where $L \triangleq N\max(A_{\max},S_{\max})^2$.
    \end{claim}
    Thus, based on Claim \ref{claim_3}, in order to establish condition (C1) for all $\epsilon \in (0,\epsilon_0)$, it suffices to show 
    \begin{equation}
    \label{eq:collapse_key}
     	\ex{\inner{\Qc(t)}{\A(t) - \s(t)} \mid \Q(t) = \Q}\le -c\norm{\Qc(t)}
    \end{equation} 
    holds for all $\epsilon \in (0,\epsilon_0)$, and $c$ is independent of $\epsilon$. 

    To this end, first recall the permutation $\sigma_t(\cdot)$ of $(1,2,\ldots,N)$ which satisfies $Q_{\sigma_t(1)}(t) \le Q_{\sigma_t(2)}(t) \le \ldots Q_{\sigma_t(N)}(t)$ and ties are broken randomly. In the following, for simplicity of notation, we let $\widehat{\Q}(t) = (Q_{\sigma_t(1)}(t),Q_{\sigma_t(2)}(t),\ldots,Q_{\sigma_t(N)}(t))$, and similarly the arrival process $\widehat{\A}(t) = (A_{\sigma_t(1)}(t),A_{\sigma_t(2)}(t),\ldots,A_{\sigma_t(N)}(t))$ and the service vector $\widehat{\s}(t) = (S_{\sigma_t(1)}(t),S_{\sigma_t(2)}(t),\ldots,S_{\sigma_t(N)}(t))$. 
	Now, the left-hand-side of Eq. \eqref{eq:collapse_key} can be written as follows.
    \begin{align}
    \label{eq:middle}
    	&\ex{\inner{\Qc(t)}{\A(t) - \s(t)} \mid \Q(t) = \Q}\nonumber\\
    	\ep{a} &\ex{\inner{\widehat{\Q}_{\perp}(t)}{\widehat{\A}(t) - \widehat{\s}(t)} \mid \Q(t) = \Q}\nonumber\\
    	\ep{b} &\sum_{n=1}^N \widehat{Q}_{\perp n}\left[\lambda_{\Sigma}\left(\Delta_n(t) + \frac{\mu_{\sigma_t(n)}}{\mu_{\Sigma}} \right) - \mu_{\sigma_t(n)} \right]\nonumber\\
    	\ep{c} & \sum_{n=1}^N \widehat{Q}_{\perp n}\Delta_n(t)\lambda_{\Sigma} + \sum_{n=1}^N \widehat{Q}_{\perp,n}\left(-\epsilon \frac{\mu_{\sigma_t(n)}}{\mu_{\Sigma}} \right)\nonumber\\
    	\le & \sum_{n=1}^N \widehat{Q}_{\perp n}\Delta_n(t)\lambda_{\Sigma} + \epsilon\norm{\widehat{\Q}_{\perp}(t)}_1
    \end{align}
    where (a) comes from the fact that the cone $\mathcal{K}_{\alpha}$ is symmetry with respect to the line $\mathbf{1} = (1,1,\ldots,1)$; In (b), $\widehat{Q}_{\perp n}$ is the $n$-th component of the vector $\widehat{\Q}_{\perp}(t)$ and (b) holds because of the definition of $\Delta(t)$ in Eq. \eqref{eq:delta_definition}, and the fact that the service process is independent of queue lengths; (c) follows from the fact that $\lambda_{\Sigma} = \mu_{\Sigma} - \epsilon$.
    % (d) holds since $\norm{\mathbf{x}}_1 \le \sqrt{N} \norm{\mathbf{x}}$ for any $\mathbf{x} \in \mathbb{R}^N$.

    Now, let us focus on the first term of Eq. \eqref{eq:middle}. To establish an upper bound on it, we will first establish the following important monotone property of $\widehat{\Q}_{\perp}(t)$. That is, 
    \begin{align}
    \label{eq:monotone}
    	% \widehat{Q}_{\perp,1}(t) \le \widehat{Q}_{\perp,2}(t) \le \cdots \le \widehat{Q}_{\perp,N}(t).
    	\widehat{Q}_{\perp1}(t) \le \widehat{Q}_{\perp2}(t) \le \cdots \le \widehat{Q}_{\perp N}(t).
    \end{align}
 	First, in the case of $\alpha = 1$, the cone $\mathcal{K}_{\alpha}$ reduces to the line $\mathbf{1}$. Thus, it can be easily obtained that $\widehat{Q}_{\perp n}(t) = Q_{\sigma_t(n)}(t) - Q_{\text{avg} }(t)$ where $Q_{\text{avg} }(t) = \sum_{n=1}^N Q_n(t)/N$, which satisfies the monotone property. Hence, we are only left with the task of establishing the monotone property for the case of $\alpha \in (0,1)$.

 	Note that, since $\widehat{\Q}_{\parallel}(t) \in \mathcal{K}_{\alpha}$, we have $\widehat{\Q}_{\parallel}(t) = \sum_{n=1}^N w_n\mathbf{b}^{(n)}$, where $w_n \ge 0$. Let $\mathcal{I}$ be a subset of $\{1,2,\ldots,N\}$ such that for any $i \in \mathcal{I}$ $w_i > 0$ and for any $i \notin \mathcal{I}$, $w_i = 0$, i.e., the subset $\mathcal{I}$ contains all the index $n$ such that $w_n > 0$. It suffices to consider the case when $\mathcal{I}$ is nonempty. This is because when $\mathcal{I}$ is empty, we have $\widehat{\Q}_{\parallel}(t) = 0$, which directly implies the monotone property since $\widehat{\Q}_{\perp}(t) = \widehat{\Q}(t)$ in this case.

 	Now consider the case when $\mathcal{I}$ is nonempty. First, we have 
 	\begin{align}
 	\label{eq:all_i}
 		\inner{\widehat{\Q}_{\perp}(t)}{\mathbf{b}^{(i)}} \le 0
 	\end{align}
 	holds for any $i \in \{1,2,\ldots,N\}$. Moreover, for any $i \in \mathcal{I}$
 	\begin{align}
 	\label{eq:i_in_I}
 		\inner{\widehat{\Q}_{\perp}(t)}{\mathbf{b}^{(i)}} = 0.
 	\end{align}
 	The inequality in Eq. \eqref{eq:all_i} follows from the fact that $\mathbf{b}^{(i)} \in \mathcal{K}_{\alpha}$ and $\widehat{\Q}_{\perp}(t) \in \mathcal{K}_{\alpha}^{\circ}$. The equality in Eq. \eqref{eq:i_in_I} follows from the fact that
 	\begin{align*}
 		0 = \inner{\widehat{\Q}_{\perp}(t)}{\widehat{\Q}_{\parallel}(t)} = \sum_{i\in \mathcal{I}}w_i\inner{\widehat{\Q}_{\perp}(t)}{\mathbf{b}^{(i)}},
 	\end{align*}
 	along with Eq. \eqref{eq:all_i} and $w_i > 0$ for all $i \in \mathcal{I}$. Eqs. \eqref{eq:all_i} and \eqref{eq:i_in_I} enable us to establish the following claim, the proof of which is relegated to Appendix \ref{sec:proof_of_Claim_4}.
 	\begin{claim}
    \label{claim_4}
 		If $m \in \mathcal{I}$ with $1 \le m <N-1$, then $m+1 \in \mathcal{I}$.
 	\end{claim}
 	% This claim can be proved by contradiction. Suppose $m \in \mathcal{I}$ but $m+1 \notin \mathcal{I}$, then by the definition of $\mathcal{I}$, we have $w_{m} >0$ and $w_{m+1} = 0$. From the definition of $\mathbf{b}^{(i)}$ and the fact that $\widehat{\Q}_{\parallel}(t) = \sum_{i \in \mathcal{I}} w_i\mathbf{b}^{(i)}$, we have 
 	% \begin{align*}
 	% 	\widehat{Q}_{\parallel m}(t) - \widehat{Q}_{\parallel m+1}(t) = w_m(1-\alpha) > 0,
 	% \end{align*}
 	% which implies that 
 	% \begin{align}
 	% \label{eq:qc_larger}
 	% 	\widehat{Q}_{\perp m}(t) < \widehat{Q}_{\perp m+1}(t),
 	% \end{align}
 	% since $\widehat{Q}_m(t) \le \widehat{Q}_{m+1}(t)$. Then, it follows that 
 	% \begin{align*}
 	% 	\inner{\widehat{\Q}_{\perp}(t)}{\mathbf{b}^{(m+1)}} &= \alpha \sum_{n=1}^N \widehat{\Q}_{\perp n}(t) + (1-\alpha)\widehat{Q}_{\perp m+1}(t)\\
 	% 	& \gp{a} \alpha \sum_{n=1}^N \widehat{\Q}_{\perp n}(t) + (1-\alpha)\widehat{Q}_{\perp m}(t)\\
 	% 	& = \inner{\widehat{\Q}_{\perp}(t)}{\mathbf{b}^{(m)}}\\
 	% 	& \ep{b} 0
 	% \end{align*}
 	% where (a) follows from Eq. \eqref{eq:qc_larger}; and (b) comes from Eq. \eqref{eq:i_in_I}. However, by Eq. \eqref{eq:all_i}, we must have $\inner{\widehat{\Q}_{\perp}(t)}{\mathbf{b}^{(m+1)}} \le 0$. Hence, Claim 4 is true.

 	Note that Claim~\ref{claim_4} directly implies that there exists an $m_0$ (which depends on $\Q(t)$) such that for all $i \ge m_0$, $i \in \mathcal{I}$ and for $i<m_0$, $i \notin \mathcal{I}$. Hence, by Eq. \eqref{eq:i_in_I}, for all $i \ge m_0$
 	\begin{align*}
 		0 = \inner{\widehat{\Q}_{\perp}(t)}{\mathbf{b}^{(i)}} = \alpha \sum_{n=1}^N \widehat{\Q}_{\perp n}(t) + (1-\alpha)\widehat{Q}_{\perp i}(t),
 	\end{align*}
 	which implies that for all $i \ge m_0$
 	\begin{align}
 	\label{eq:larger_m_0}
 		\widehat{\Q}_{\perp i}(t) = c \ge 0 
 	\end{align}
 	for some constant $c$. This is because $\sum_{n=1}^N \widehat{\Q}_{\perp n}(t) \le 0$, due to the fact that $\mathbf{1} \in \mathbf{K}_{\alpha}$ and $\widehat{\Q}_{\perp}(t) \in \mathcal{K}_{\alpha}^{\circ}$. On the other hand, for any $i\le j < m_0$, we have 
 	\begin{align}
 	\label{eq:less_m_0}
 		\widehat{\Q}_{\perp i}(t) \le \widehat{\Q}_{\perp j}(t).
 	\end{align}
 	This holds since $\widehat{\Q}_{\parallel i}(t) = \widehat{\Q}_{\parallel j}(t)$ and $\widehat{Q}_i(t) \le \widehat{Q}_j(t)$. Moreover, we have 
 	\begin{align}
 	\label{eq:cross_m_0}
 		\widehat{\Q}_{\perp m_0}(t) \ge \widehat{\Q}_{\perp (m_0-1)}(t).
 	\end{align}
 	This can be shown by contradiction. Suppose $\widehat{\Q}_{\perp (m_0-1)}(t) >\widehat{\Q}_{\perp m_0}(t)$, then 
 	\begin{align*}
 		\inner{\widehat{\Q}_{\perp}(t)}{\mathbf{b}^{(m_0 -1)}} &= \alpha \sum_{n=1}^N \widehat{\Q}_{\perp n}(t) + (1-\alpha)\widehat{Q}_{\perp (m_0-1)}(t)\\
 		& > \alpha \sum_{n=1}^N \widehat{\Q}_{\perp n}(t) + (1-\alpha)\widehat{Q}_{\perp m_0}(t)\\
 		& = \inner{\widehat{\Q}_{\perp}(t)}{\mathbf{b}^{(m_0)}}\\
 		& =  0
 	\end{align*}
 	which contradicts with Eq.\eqref{eq:all_i}. Then, combining Eqs. \eqref{eq:larger_m_0}, \eqref{eq:less_m_0} and \eqref{eq:cross_m_0}, yields the fact that  $\widehat{\Q}_{\perp N}(t) \ge 0$ and the monotone property in Eq. \eqref{eq:monotone}. As a result, we have $\widehat{\Q}_{\perp 1}(t) \le 0 $ since otherwise $\sum_{n=1}^N \widehat{\Q}_{\perp,n}(t)$ would be strictly positive.

 	Having established the monotone property of $\widehat{\Q}_{\perp}(t)$ and auxiliary results that $\widehat{\Q}_{\perp N}(t) \ge 0$ and $\widehat{\Q}_{\perp 1}(t) \le 0 $, we can now proceed to obtain an upper bound on the first term in Eq. \eqref{eq:middle}. In particular, we can first bound it in terms of $|\widehat{\Q}_{\perp1}(t)|$ and the $\delta$ in Eq. \eqref{eq:condition_2}. In particular, we have 
 	\begin{align}
 	\label{eq:cross_upper_delta_qc}
 		\sum_{n=1}^N \widehat{Q}_{\perp n}\Delta_n(t)\lambda_{\Sigma} \le -\lambda_{\Sigma}\delta |\widehat{\Q}_{\perp1}(t)|.
 	\end{align}
 	This upper bound can be verified as follows. First, if $\Q(t) \in \mathcal{K}_{\alpha}$, then $\widehat{\Q}_{\perp n}(t) = 0$ for all $n$, and hence Eq. \eqref{eq:cross_upper_delta_qc} holds. If $\Q(t) \notin \mathcal{K}_{\alpha}$, then $\Delta(t)$ satisfies the two conditions in  Eqs. \eqref{eq:condition_1} and \eqref{eq:condition_2} in Theorem \ref{thm:theorem_2}, which specify the construction process of $\Delta(t)$. In particular, each $\Delta(t)$ that satisfies the two conditions can be constructed as follows. To begin with, all the $\Delta_n(t)$ is $0$. Then, according to the condition in Eq. \eqref{eq:condition_2}, we should first decrease $\Delta_N(t)$ by the amount of $\delta$, and increase $\Delta_1(t)$ by the amount of $\delta$. After this, the left-hand-side of Eq. \eqref{eq:cross_upper_delta_qc} is equal to $\lambda_{\Sigma} \left(\delta\widehat{\Q}_{\perp 1}(t) + (-\delta)\widehat{\Q}_{\perp N}(t)\right)$, which is upper bounded by $-\lambda_{\Sigma}\delta |\widehat{\Q}_{\perp 1}(t)|$, since $\widehat{\Q}_{\perp N}(t) \ge 0$ and $\widehat{\Q}_{\perp 1}(t) \le 0$. Next, due to the condition in Eq.\eqref{eq:condition_1} and the fact that $\sum_{n=1}^N \Delta_n(t) = 0$, any further procedure (if needed) for the construction of $\Delta(t)$ can only take the following way: it decreases some $\Delta_i(t)$ by a certain amount (say $c_1$) where $i \ge k$, and then increase some $\Delta_j(t)$ by the same amount $c_1$ where $j\le k$. We claim that any of this procedure cannot increase the value of the left-hand-side of Eq. \eqref{eq:cross_upper_delta_qc} due to the monotone property of $\widehat{\Q}_{\perp}(t)$. To see this, let us denote by $\mathcal{L}_{ij}$ the change of the value of the left-hand-side of Eq. \eqref{eq:cross_upper_delta_qc} incurred by the procedure above. Thus, we have 
 	\begin{align*}
 		\mathcal{L}_{ij} = -c_1\widehat{Q}_{\perp i}(t) + c_1 \widehat{Q}_{\perp j}(t) \le 0,
 	\end{align*}
 	which follows from the monotone property of $\widehat{\Q}_{\perp}(t)$. Therefore, we have verified the upper bound in Eq. \eqref{eq:cross_upper_delta_qc}.

 	Next, we establish an upper bound on $\Vert\widehat{\Q}_{\perp}(t)\rVert_1$ in terms of $|\widehat{\Q}_{\perp 1}(t)|$ as follows 
 	\begin{align}
 	\label{eq:cross_upper_delta_qc_2}
 		 \norm{\widehat{\Q}_{\perp}(t)}_1 \le 2 N |\widehat{\Q}_{\perp 1}(t)|.
 	\end{align}
 	This follows from the monotone property of $\widehat{\Q}_{\perp}(t)$ and the fact that $\sum_{n=1}^N \widehat{\Q}_{\perp n}(t) \le 0$. Now, combining Eqs. \eqref{eq:middle}, \eqref{eq:cross_upper_delta_qc} and \eqref{eq:cross_upper_delta_qc_2}, we obtain that 
 	\begin{align}
 	\label{eq:drift_in_delta}
 		&\ex{\inner{\Qc(t)}{\A(t) - \s(t)} \mid \Q(t) = \Q}\nonumber\\
 		 \le & \left(\epsilon - \frac{\lambda_{\Sigma}\delta}{2N}\right)\norm{\widehat{\Q}_{\perp}(t)}_1\nonumber\\
 		 \le & -\frac{\mu_{\Sigma}\delta}{4N} \norm{\widehat{\Q}_{\perp}(t)}_1 \text{ whenever } \epsilon \le \frac{\mu_{\Sigma}\delta}{4N+2\delta}\nonumber\\
 		 \le & -\frac{\mu_{\Sigma}\delta}{4N} \norm{{\Q}_{\perp}(t)}
 	\end{align}
 	where the last inequality comes from the fact that $\lVert\widehat{\Q}_{\perp}(t)\rVert_1 = \norm{{\Q}_{\perp}(t)}_1$ and  $\norm{\mathbf{x}}_1 \ge \norm{\mathbf{x}}$ for any $\mathbf{x} \in \mathbb{R}^N$. This establishes the inequality in Eq. \eqref{eq:collapse_key} with $c = {\mu_{\Sigma}\delta}/{4N}$ and $\epsilon_0 ={\mu_{\Sigma}\delta}/{(4N+2\delta)}.$ Hence, based on Claim \ref{claim_3}, we have verified the condition (C1) in Lemma \ref{lem:basis}, which directly establishes the state-space collapse result in Proposition \ref{prop:collapse}.
    \end{proof}
\section{Conclusions}
% Heavy-traffic analysis has been an important tool for characterizing the delay performance of load balancing systems. It is often believed that heavy-traffic delay optimality implies good delay performance in practice. However, in this paper, we show that this conventional wisdom is not necessarily true. 

{We have rigorously shown that even under a multi-dimensional state-space collapse, steady-state heavy-traffic delay optimality can be achieved for a general load balancing system. This result suggests that the insight behind heavy-traffic optimality conveyed by diffusion approximations is still valid in \emph{steady state}, thus complementing and extending the diffusion approximation results in \cite{kelly1993dynamic},~\cite{teh2002critical}. Moreover, our steady-state delay optimality result might also give a possible direction for proving the interchange of limits for the diffusion approximation results in \cite{kelly1993dynamic},~\cite{teh2002critical}. 
By leveraging this result, we are able to explore the greater flexibility provided by allowing a multi-dimensional state-space collapse in designing new load balancing policies that are both throughput optimal and heavy-traffic delay optimal in steady state. Furthermore, the proof techniques used in this paper are of independent interest as well.}

\bibliographystyle{plain}
\bibliography{ref} 

\begin{thebibliography}{10}

\bibitem{armony2005dynamic}
Mor Armony.
\newblock Dynamic routing in large-scale service systems with heterogeneous
  servers.
\newblock {\em Queueing Systems}, 51(3-4):287--329, 2005.

\bibitem{bell2001dynamic}
Steven~L Bell and Ruth~J Williams.
\newblock Dynamic scheduling of a system with two parallel servers in heavy
  traffic with resource pooling: asymptotic optimality of a threshold policy.
\newblock {\em Annals of Applied Probability}, pages 608--649, 2001.

\bibitem{bramson1998state}
Maury Bramson.
\newblock State space collapse with application to heavy traffic limits for
  multiclass queueing networks.
\newblock {\em Queueing Systems}, 30(1-2):89--140, 1998.

\bibitem{chen2013fundamentals}
Hong Chen and David~D Yao.
\newblock {\em Fundamentals of queueing networks: Performance, asymptotics, and
  optimization}, volume~46.
\newblock Springer Science \& Business Media, 2013.

\bibitem{chen2012asymptotic}
Hong Chen and Heng-Qing Ye.
\newblock Asymptotic optimality of balanced routing.
\newblock {\em Operations research}, 60(1):163--179, 2012.

\bibitem{dai2011state}
JG~Dai and Tolga Tezcan.
\newblock State space collapse in many-server diffusion limits of parallel
  server systems.
\newblock {\em Mathematics of Operations Research}, 36(2):271--320, 2011.

\bibitem{eryilmaz2012asymptotically}
Atilla Eryilmaz and R~Srikant.
\newblock Asymptotically tight steady-state queue length bounds implied by
  drift conditions.
\newblock {\em Queueing Systems}, 72(3-4):311--359, 2012.

\bibitem{foschini1978basic}
G~Foschini and JACK Salz.
\newblock A basic dynamic routing problem and diffusion.
\newblock {\em IEEE Transactions on Communications}, 26(3):320--327, 1978.

\bibitem{foster2008cloud}
Ian Foster, Yong Zhao, Ioan Raicu, and Shiyong Lu.
\newblock Cloud computing and grid computing 360-degree compared.
\newblock In {\em 2008 Grid Computing Environments Workshop}, pages 1--10.
  Ieee, 2008.

\bibitem{george2011hbase}
Lars George.
\newblock {\em HBase: the definitive guide}.
\newblock " O'Reilly Media, Inc.", 2011.

\bibitem{gupta2007analysis}
Varun Gupta, Mor~Harchol Balter, Karl Sigman, and Ward Whitt.
\newblock Analysis of join-the-shortest-queue routing for web server farms.
\newblock {\em Performance Evaluation}, 64(9):1062--1081, 2007.

\bibitem{gurvich2009queue}
Itay Gurvich and Ward Whitt.
\newblock Queue-and-idleness-ratio controls in many-server service systems.
\newblock {\em Mathematics of Operations Research}, 34(2):363--396, 2009.

\bibitem{harrison1998heavy}
J~Michael Harrison.
\newblock Heavy traffic analysis of a system with parallel servers: asymptotic
  optimality of discrete-review policies.
\newblock {\em Annals of applied probability}, pages 822--848, 1998.

\bibitem{kang2009state}
WN~Kang, FP~Kelly, NH~Lee, RJ~Williams, et~al.
\newblock State space collapse and diffusion approximation for a network
  operating under a fair bandwidth sharing policy.
\newblock {\em The Annals of Applied Probability}, 19(5):1719--1780, 2009.

\bibitem{kang2012diffusion}
WN~Kang and RJ~Williams.
\newblock Diffusion approximation for an input-queued packet switch operating
  under a maximum weight algorithm.
\newblock {\em Stochastic Systems}, 2012.

\bibitem{kelly1993dynamic}
FP~Kelly and CN~Laws.
\newblock Dynamic routing in open queueing networks: Brownian models, cut
  constraints and resource pooling.
\newblock {\em Queueing systems}, 13(1-3):47--86, 1993.

\bibitem{maguluri2018optimal}
Siva~Theja Maguluri, Sai~Kiran Burle, and R~Srikant.
\newblock Optimal heavy-traffic queue length scaling in an incompletely
  saturated switch.
\newblock {\em Queueing Systems}, 88(3-4):279--309, 2018.

\bibitem{maguluri2016heavy}
Siva~Theja Maguluri, R~Srikant, et~al.
\newblock Heavy traffic queue length behavior in a switch under the maxweight
  algorithm.
\newblock {\em Stochastic Systems}, 6(1):211--250, 2016.

\bibitem{maguluri2014heavy}
Siva~Theja Maguluri, R~Srikant, and Lei Ying.
\newblock Heavy traffic optimal resource allocation algorithms for cloud
  computing clusters.
\newblock {\em Performance Evaluation}, 81:20--39, 2014.

\bibitem{paschos2016routing}
Georgios~S Paschos, Mathieu Leconte, and Apostolos Destounis.
\newblock Routing with blinkers: Online throughput maximization without queue
  length information.
\newblock In {\em Information Theory (ISIT), 2016 IEEE International Symposium
  on}, pages 1436--1440. IEEE, 2016.

\bibitem{reiman1984some}
Martin~I Reiman.
\newblock Some diffusion approximations with state space collapse.
\newblock In {\em Modelling and performance evaluation methodology}, pages
  207--240. Springer, 1984.

\bibitem{stolyar2004maxweight}
Alexander~L Stolyar et~al.
\newblock Maxweight scheduling in a generalized switch: State space collapse
  and workload minimization in heavy traffic.
\newblock {\em The Annals of Applied Probability}, 14(1):1--53, 2004.

\bibitem{teh2002critical}
Yih-Choung Teh and Amy~R Ward.
\newblock Critical thresholds for dynamic routing in queueing networks.
\newblock {\em Queueing Systems}, 42(3):297--316, 2002.

\bibitem{wangIFIP}
Weina Wang, Siva~Theja Maguluri, R~Srikant, and Lei Ying.
\newblock Heavy-traffic delay insensitivity in connection-level models of data
  transfer with proportionally fair bandwidth sharing.
\newblock 2017.

\bibitem{wang2018heavy}
Weina Wang, Siva~Theja Maguluri, R~Srikant, and Lei Ying.
\newblock Heavy-traffic delay insensitivity in connection-level models of data
  transfer with proportionally fair bandwidth sharing.
\newblock {\em ACM SIGMETRICS Performance Evaluation Review}, 45(2):232--245,
  2018.

\bibitem{wang2016maptask}
Weina Wang, Kai Zhu, Lei Ying, Jian Tan, and Li~Zhang.
\newblock Maptask scheduling in mapreduce with data locality: Throughput and
  heavy-traffic optimality.
\newblock {\em IEEE/ACM Transactions on Networking}, 24(1):190--203, 2016.

\bibitem{xie2015priority}
Qiaomin Xie and Yi~Lu.
\newblock Priority algorithm for near-data scheduling: Throughput and
  heavy-traffic optimality.
\newblock In {\em Proceedings of IEEE International Conference on Computer
  Communications (INFOCOM)}, pages 963--972, 2015.

\bibitem{xie2016scheduling}
Qiaomin Xie, Ali Yekkehkhany, and Yi~Lu.
\newblock Scheduling with multi-level data locality: Throughput and
  heavy-traffic optimality.
\newblock In {\em Proceedings of IEEE International Conference on Computer
  Communications (INFOCOM)}, pages 1--9, 2016.

\bibitem{zaharia2010delay}
Matei Zaharia, Dhruba Borthakur, Joydeep Sen~Sarma, Khaled Elmeleegy, Scott
  Shenker, and Ion Stoica.
\newblock Delay scheduling: a simple technique for achieving locality and
  fairness in cluster scheduling.
\newblock In {\em Proceedings of the 5th European conference on Computer
  systems}, pages 265--278. ACM, 2010.

\bibitem{zhou2017designing}
Xingyu Zhou, Fei Wu, Jian Tan, Yin Sun, and Ness Shroff.
\newblock Designing low-complexity heavy-traffic delay-optimal load balancing
  schemes: Theory to algorithms.
\newblock {\em arXiv preprint arXiv:1710.04357}, 2017.

\end{thebibliography}

% !TEX root = ./sig2018_winter.tex
% \section{Appendix}
\appendix
\section{Proof of Lemma \ref{claim_1}}
\label{sec:proof_of_claim_1}
\begin{proof}

  Let us consider the Lyapunov function $V_1(\Q(t)) \triangleq \norm{\Q(t)}_1^2$, and the corresponding conditional mean drift is given by 
  \begin{align}
  \label{eq:heavy}
    &\ex{V_1(\Q(t+1)) - V_1(\Q(t)) \mid \Q(t) = \Q}\nonumber\\
    = & \ex{ \norm{\Q(t+1)}_1^2 - \norm{\Q(t)}_1^2 \mid \Q(t) = \Q    }\nonumber\\
    % = & \ex{\left(\norm{\Q(t)}_1+\norm{\A(t)}_1-\norm{\s(t)}_1 + \norm{\UU(t)}_1 \right)^2  \mid \Q(t) = \Q    }\nonumber\\
    % & - \ex{\norm{\Q(t)}_1^2 \mid \Q(t) = \Q   }\nonumber\\
        = & \mathbb{E}\left[ 2\norm{\Q{}}_1\left( \norm{\A}_1 - \norm{\s}_1\right) + \left( \norm{\A}_1-\norm{\s}_1\right)^2\right.\nonumber\\
         & +\left.{} 2\left(\norm{\Q}_1 +\norm{\A}_1 - \norm{\s}_1 \right) \norm{\UU}_1 + \norm{\UU}_1^2 \vphantom{\left( \norm{\A}_1-\norm{\s}_1\right)^2}      \mid \Q(t) = \Q   \right]\nonumber\\
        = & \mathbb{E}\left[ 2\norm{\Q}_1\left( \norm{\A}_1 - \norm{\s}_1\right) + \left( \norm{\A}_1-\norm{\s}_1\right)^2\right.\nonumber\\
         & +\left.{} 2\norm{\Q(t+1)}_1 \norm{\UU}_1 - \norm{\UU}_1^2 \vphantom{\left( \norm{\A}_1-\norm{\s}_1\right)^2}      \mid \Q(t) = \Q   \right].
  \end{align}
  By the definition of throughput optimality in this paper, we have $\ex{V_1(\overline{\Q})}$ is finite. Therefore, the mean drift of $V_1(.)$ is zero in steady-state. Taking expectation of both sides of Eq. \eqref{eq:heavy} with respect to the steady-state distribution $\overline{\Q}^{(\epsilon)}$, yields
  \begin{equation*}
  % \label{eq:sum_queuelength_equation}
    \epsilon \ex{\sum_{n=1}^{N}\overline{Q}_n^{(\epsilon)}} = \frac{\zeta^{(\epsilon)}}{2} + \ex{\big\lVert\overline{\Q}^{(\epsilon)}(t+1)\big\rVert_1 \big\lVert\overline{\UU}^{(\epsilon)} (t)\big\rVert_1} - \frac{1}{2}\ex{\big\lVert\overline{\UU}^{(\epsilon)}\big\rVert_1^2}
  \end{equation*}
  where $\zeta^{(\epsilon)} = (\sigma_\Sigma^{(\epsilon)})^2 + \nu_\Sigma^2 + \epsilon^2$. Then by utilizing the property of unused service shown in Lemma \ref{lem:unused_service}, we have 
  \begin{equation*}
    \frac{\zeta^{(\epsilon)}}{2} + {\mathcal{T}}^{(\epsilon)} - \frac{1}{2}c_1\epsilon \le \epsilon \ex{\sum_{n=1}^{N}\overline{Q}_n^{(\epsilon)}} \le \frac{\zeta^{(\epsilon)}}{2} + {\mathcal{T}}^{(\epsilon)},
  \end{equation*}
  in which ${\mathcal{T}}^{(\epsilon)} = \ex{\big\lVert\overline{\Q}^{(\epsilon)}(t+1)\big\rVert_1 \norm{\overline{\UU}^{(\epsilon)} (t)}_1}$. Since $\zeta^{(\epsilon)}$ converges to $\zeta$, from the inequality above and the definition of heavy-traffic delay optimality, we can easily see that the sufficient and necessary condition is $\lim_{\epsilon \downarrow 0} {\mathcal{T}}^{(\epsilon)} = 0$, which completes the proof of Lemma \ref{claim_1}. 
\end{proof}

\section{Proof of Lemma \ref{lemma:fluid model}}
\label{appendex:proof_fluid_model}
\begin{proof}
  (a) To show Eqs. \eqref{eq:sigma} and \eqref{eq:s} hold,  we can directly apply FSLLN (functional strong law of large numbers) to obtain that $\frac{1}{x_{n_k}}{\mathcal{A} }_{\Sigma}^{(x_{n_k})}(x_{n_k} t) \to \lambda t$ and $\frac{1}{x_{n_k}}{\mathcal{S} }_{i}^{(x_{n_k})}(x_{n_k} t) \to \mu_i t$, u.o.c, and each limiting function is Lipschitz continuous. 

  (b) For Eqs. \eqref{eq:g} and \eqref{eq:b}, for any given $0 \le t_1  \le t_2$, we have 
  \begin{equation}
  \label{eq:delta_g}
    0 \le \frac{1}{x_{n_k}} \left[ {\mathcal{G} }_{i}^{(x_{n_k})}(x_{n_k} t_2) - {\mathcal{G} }_{i}^{(x_{n_k})}(x_{n_k} t_1) \right] \le (t_2 - t_1).
  \end{equation}
  Therefore, the sequence of functions $\{\frac{1}{x_{n_k}}{\mathcal{G} }_{i}^{(x_{n_k})}(x_{n_k} t) \}$ is uniformly bounded and uniformly equicontinuous. As a result, by the Arzela-Ascoli theorem, there must exist a subsequence along which Eq. \eqref{eq:g} hold. In addition, Eq. \eqref{eq:delta_g} also implies that each limiting function $g_i$ is Lipschitz continuous. Same argument can be used to show Eq. \eqref{eq:b} hold and each limiting function $b_i$ is Lipschitz continuous.

  (c) For Eq. \eqref{eq:qq}, since the sequence $\{\frac{1}{x_n}Q_i^{x_n}(0)\}$ is upper bounded by 1 as a result of Eq. \eqref{eq:norm}, we have that there is a subsequence such that $\frac{1}{x_{n_k}}{Q}_i^{(x_{n_k})}(0) \to q_i(0)$. Hence, the convergence of Eq. \eqref{eq:qq} follows directly from Eq. \eqref{eq:q_dynamic}, and each limiting function $q_i$ is Lipschitz continuous.

  (d) To show that Eqs. \eqref{eq:ai} and \eqref{eq:di} hold, we utilize the fact that the arrival and departure process are bounded. Take the arrival process for example, we have 
  \begin{equation*}
  \label{eq:delta_ai}
    0 \le \frac{1}{x_{n_k}} \left[ {\mathcal{A} }_{i}^{(x_{n_k})}(x_{n_k} t_2) - {\mathcal{A} }_{i}^{(x_{n_k})}(x_{n_k} t_1) \right] \le A_{\max}(t_2 - t_1),
  \end{equation*}
  where $A_{\max}$ is the maximum number of exogenous arrivals at each time-slot. For each server $i$, we also have
  \begin{equation*}
  \label{eq:delta_di}
    0 \le \frac{1}{x_{n_k}} \left[ {\mathcal{D} }_{i}^{(x_{n_k})}(x_{n_k} t_2) - {\mathcal{D} }_{i}^{(x_{n_k})}(x_{n_k} t_1) \right] \le S_{\max}(t_2 - t_1),
  \end{equation*}
  where $S_{\max}$ is the maximum number of offered service at each time-slot. As a result, with the similar argument as in Eq. \eqref{eq:delta_g}, we can easily show Eqs. \eqref{eq:ai} and \eqref{eq:di} hold, and each limiting function is Lipschitz continuous. 
  % The other method doesn't rely on the fact of bounded arrivals and service, it is based on Theorem 5.3 in \cite{chen2013fundamentals}, which is repeated  in the next Theorem \ref{thm:random-time change}. We defer the second method to the following Lemma \ref{lemma:fluid limit}, which gives equations that each limiting function should satisfy.
\end{proof}

\section{Proof of Lemma \ref{lemma:fluid limit}}
\label{appendex:proof_fluid_limit}
In the proof, we will utilize the random time-change theorem in Chapter 5 of \cite{chen2013fundamentals}, which is presented below for easy reference.

\begin{theorem}[Random Time-Change Theorem]
\label{thm:random-time change}
Let $\{ X_n, n \ge 1\}$ and $\{ Y_n, n \ge 1\}$ be two sequences in ${D}^J$ (i.e., the space of $J$-dimensional real-valued functions on $[0,\infty)$ that are right-continuous and with left limits.). Assume that $Y_n$ is nondecreasing with $Y_n(0) = 0$. If as $n \to \infty$, $(X_n,Y_n)$ converges uniformly on compact sets to $(X,Y)$ with $X$ and $Y$ in ${C}^J$ (i.e., the space of $J$-dimensional real-valued continuous functions on $[0,\infty)$), then $X_n(Y_n)$ converges uniformly on compact sets to $X(Y)$, where $X_n(Y_n) = X_n \circ Y_n = \{X_n(Y_n(t)), t \ge 0 \}$ and  $X(Y) = X \circ Y = \{X(Y(t)), t \ge 0 \}$.
\end{theorem}

Now, we present the proof of Lemma \ref{lemma:fluid limit}.
\begin{proof}
  (a) Eqs. \eqref{eq:fluid_a_all} and \eqref{eq:fluid_s} directly follows from FSLLN under our assumptions for the exogenous arrival process and each service process.

  (b) Eq. \eqref{eq:fluid_q} follows from the definition of the queue length dynamic. Eq. \eqref{eq:sum_g} follows from the definition directly.

  (c) Eqs. \eqref{eq:fluid_ai} and \eqref{eq:fluid_di} are the results of Theorem \ref{thm:random-time change}. More specifically, let $X_n = \frac{1}{x_{n_k}}{\mathcal{A} }_{\Sigma}^{(x_{n_k})}(x_{n_k} t)$ and $Y_n = \frac{1}{x_{n_k}}{\mathcal{G} }_{i}^{(x_{n_k})}(x_{n_k} t)$ and we have $(X_n,Y_n) \to (\lambda t, g_i(t))$ uniformly on compact set, and $Y_n$ is nondecreasing with $Y_n(0) = 0$. Thus by Theorem \ref{thm:random-time change}, we have $X_n(Y_n(t)) = \frac{1}{x_{n_k}}{\mathcal{A} }_i^{(x_{n_k})}(x_{n_k} t) \to X(Y(t)) = \lambda g_i(t) = a_i(t)$ uniformly on compact sets. Similar argument can be used to show Eq. \eqref{eq:fluid_di} hold.

  (d) Note that $t$ is the regular time in Eq. \eqref{eq:derivative}, and hence $q_i^{\prime}(t)$ is well defined and exists. Therefore, the left-derivative $q_i^{\prime}(t-)$ should be equal to the right-derivative $q_i^{\prime}(t+)$. By the non-negativity of $q_i(t)$, if $q_i(t) = 0$, then we must have $q_i^{\prime}(t-) \le 0$ and $q_i^{\prime}(t+) \ge 0$, which results in $q_i^{\prime}(t) = 0$.
    For the case $q_i(t) > 0$, we need to show $d_i^{\prime}(t) = \mu_i$. It suffices to consider the right-derivative as it is equal to the derivative at a regular time $t$. Suppose $q_i(t) > 0$, then by the continuity of $q_i(t)$, there exists a $\delta > 0$, such that $a = \min_{t_s\in[t,t+\delta]}q(t_s) > 0$. Therefore, for sufficient large $x_{n_k}$, we have
    \begin{equation}
      \frac{1}{x_{n_k}}{Q}_i^{(x_{n_k})}(x_{n_k} t_s) \ge \frac{a}{2}, \text{for any  } t_s\in[t,t+\delta] \text{ and } \frac{a}{2}x_{n_k} \ge 1,
    \end{equation}
    which implies that ${Q}_i^{(x_{n_k})}(x_{n_k} t) \ge 1$ for any $t_s \in [t,t+\delta]$. Therefore, we have
    \begin{align*}
      &\frac{1}{x_{n_k}}{\mathcal{D} }_i^{(x_{n_k})}(x_{n_k} t_s) - \frac{1}{x_{n_k}}{\mathcal{D} }_i^{(x_{n_k})}(x_{n_k} t) \\
      = &\frac{1}{x_{n_k}}{\mathcal{S} }_i^{(x_{n_k})}(x_{n_k} t_s) - \frac{1}{x_{n_k}}{\mathcal{S} }_i^{(x_{n_k})}(x_{n_k} t).
    \end{align*}
    Then, according to the definition of derivative, we have 
    \begin{equation}
    \begin{split}
      d_i^{\prime}(t) &= \lim_{t_s\to t}\lim_{x_{n_k} \to \infty} \frac{1}{x_{n_k}}\frac{{\mathcal{D} }_i^{(x_{n_k})}(x_{n_k} t_s) - {\mathcal{D} }_i^{(x_{n_k})}(x_{n_k} t)}{t_s - t}\\
      & = \lim_{t_s\to t}\lim_{x_{n_k} \to \infty} \frac{1}{x_{n_k}}\frac{{\mathcal{S} }_i^{(x_{n_k})}(x_{n_k} t_s) - {\mathcal{S} }_i^{(x_{n_k})}(x_{n_k} t)}{t_s - t}\\
      & = \mu_i
    \end{split}
    \end{equation}
    As a result, Eq. \eqref{eq:derivative} is true for any regular time $t$.
\end{proof}

\section{Proof of Claim \ref{claim_2}}
\label{sec:proof_of_Claim_2}
\begin{proof}
Since $q_{\sigma_t(m+1)}(t) > 0$, $q_{\sigma_t(m)}(t) = 0$ and both functions are continuous, we can choose a $\tau$ such that $a/b > 4/\alpha$ where $a = \min_{t_s\in[t-\tau,t+\tau]}q_{\sigma_t(m+1)}(t_s)$ and $b = \max_{t_s\in[t-\tau,t+\tau]}q_{\sigma_t(m)}(t_s)$. By the u.o.c convergence, for sufficient large $x_{n_k}$, we have for any $t_s\in[t-\tau,t+\tau]$
  \begin{equation}
      \frac{1}{x_{n_k}}{Q}_{\sigma_t(m+1)}^{(x_{n_k})}(x_{n_k} t_s) \ge \frac{a}{2} \text{ and } \frac{1}{x_{n_k}}{Q}_{\sigma_t(m)}^{(x_{n_k})}(x_{n_k} t_s) \le 2b,
    \end{equation}
    which implies that ${{Q}_{\sigma_t(m+1)}^{(x_{n_k})}(x_{n_k} t_s)}\big/{{Q}_{\sigma_t(m)}^{(x_{n_k})}(x_{n_k} t_s)} > 1/\alpha$, for any $t_s\in [t-\tau,t+\tau]$. This indicates that in the interval $[(t-\tau)x_{n_k} + 1,(t+\tau)x_{n_k}-1]$, the queue-length state is outside the cone $\mathcal{K}_\alpha$. In this case, according to the conditions for the load balancing policy in Theorem \ref{thm:theorem_2}, we have 
    \begin{align*}
      \sum_{n=m+1}^N & \left(\frac{1}{x_{n_k}}{\mathcal{G} }_{\sigma_t(n)}^{(x_{n_k})}(x_{n_k} (t+\frac{\tau}{2})) - \frac{1}{x_{n_k}}{\mathcal{G} }_{\sigma_t(n)}^{(x_{n_k})}(x_{n_k} (t-\frac{\tau}{2}))\right)\\
      & = \tau \sum_{n=m+1}^N \left(\Delta_n(t) + \frac{\mu_{\sigma_t(n)}}{\mu_{\Sigma}}\right).
    \end{align*}
    By letting $x_{n_k} \to \infty$ and from Eq.\eqref{eq:g}, we have 
    \begin{equation*}
      \sum_{n=m+1}^N  \left(g_{\sigma_t(n)}(t+\frac{\tau}{2}) - g_{\sigma_t(n)}(t+\frac{\tau}{2})\right) = \tau \sum_{n=m+1}^N \left(\Delta_n(t) + \frac{\mu_{\sigma_t(n)}}{\mu_{\Sigma}}\right),
    \end{equation*}
    which directly implies the required result of Claim \ref{claim_2}.
\end{proof}

\section{Proof of Claim \ref{claim_3}}
\label{sec:proof_of_Claim_3}
\begin{proof}
  First, we have the following bound 
  \begin{align}
  \label{eq:root_bound}
    &\ex{\Delta V_\perp(\Q) \mid \Q(t) = \Q }\nonumber \\
    \le &\frac{1}{2\norm{\Qc}}\ex{\Delta W(\Q)  - \Delta W_{\parallel}(\Q)\mid \Q(t) = \Q}.
  \end{align}
  Similar to Lemma 7 in \cite{eryilmaz2012asymptotically}, this bound directly follows from the concavity of root function and  Pythagorean theorem. Next, we will bound each term in Eq. \eqref{eq:root_bound}, respectively. To begin with, we have an upper bound for the first term as follows.
  \begin{align}
  \label{eq:root_bound_W}
    &\ex{\Delta W(\Q) \mid \Q(t) =  \Q}\nonumber\\
    = &\ex{\norm{\Q(t) + \A(t)-\s(t) +\UU(t)}^2- \norm{\Q(t)}^2 \mid \Q}\nonumber\\
    \lep{a} &\ex{\norm{\Q(t) + \A(t)-\s(t)}^2- \norm{\Q(t)}^2 \mid \Q}\nonumber\\
    = &\ex{2\inner{\Q(t)}{\A(t)-\s(t)} + \norm{\A(t)-\s(t)}^2 \mid \Q}\nonumber\\
    \lep{b}& \ex{2\inner{\Q(t)}{\A(t)-\s(t)}\mid \Q }+ L
  \end{align}
  where (a) holds as $[\max(a,0)]^2 \le a^2 $ for any $a\in \mathbb{R}$; in (b), $L \triangleq N\max(A_{\max},S_{\max})^2$, which follows from the assumptions that  $A_\Sigma(t) \le A_{\max}$ and $S_n(t) \le S_{\max}$ for all $t \ge 0$ and all $ 1\le n \le N$, and the fact that they are both independent of the queue lengths. 

  We now turn to provide a lower bound on the second term in Eq. \eqref{eq:root_bound} as follows.
  \begin{align}
  \label{eq:root_bound_Wp}
        &\ex{\Delta W_\parallel(\Q) \mid \Q(t) =  \Q}\nonumber\\ 
        = & \ex{ 2\inner{\Qp(t)}{\Qp(t+1) - \Qp(t)} + \norm{\Qp(t+1) - \Qp(t)}^2\mid \Q}\nonumber\\
        \ge & \ex{2\inner{\Qp(t)}{\Qp(t+1) - \Qp(t)} \mid \Q}\nonumber\\
        = &2 \ex{\inner{\Qp(t)}{\Q(t+1) -\Q(t)} - 2\inner{\Qp(t)}{\Qc(t+1) - \Qc(t)} \mid \Q}\nonumber\\
        \gep{a} & \ex{2 \inner{\Qp(t)}{\Q(t+1) -\Q(t)}\mid \Q } \nonumber\\
        \gep{b} & \ex{2 \inner{\Qp(t)}{\A(t)-\s(t)} \mid \Q}
  \end{align}
  where (a) holds because $\inner{\Qp(t)}{\Qc(t)} = 0$ and $\inner{\Qc(t+1)}{\Qp(t)} \le 0$, since $\Qp(t) \in \mathcal{K}_{\alpha}$ and $\Qc(t+1)\in \mathcal{K}_{\alpha}^{\circ}$; (b) follows from the fact that all the components of $\Qp(t)$ and $\UU(t)$ are nonnegative.
    Thus, substituting Eqs. \eqref{eq:root_bound_W} and \eqref{eq:root_bound_Wp} into Eq. \eqref{eq:root_bound}, yields the bound in Claim \ref{claim_3}.
\end{proof}

\section{Proof of Claim \ref{claim_4}}
\label{sec:proof_of_Claim_4}
\begin{proof}
  This claim can be proved by contradiction. Suppose $m \in \mathcal{I}$ but $m+1 \notin \mathcal{I}$, then by the definition of $\mathcal{I}$, we have $w_{m} >0$ and $w_{m+1} = 0$. From the definition of $\mathbf{b}^{(i)}$ and the fact that $\widehat{\Q}_{\parallel}(t) = \sum_{i \in \mathcal{I}} w_i\mathbf{b}^{(i)}$, we have 
  \begin{align*}
    \widehat{Q}_{\parallel m}(t) - \widehat{Q}_{\parallel m+1}(t) = w_m(1-\alpha) > 0,
  \end{align*}
  which implies that 
  \begin{align}
  \label{eq:qc_larger}
    \widehat{Q}_{\perp m}(t) < \widehat{Q}_{\perp m+1}(t),
  \end{align}
  since $\widehat{Q}_m(t) \le \widehat{Q}_{m+1}(t)$. Then, it follows that 
  \begin{align*}
    \inner{\widehat{\Q}_{\perp}(t)}{\mathbf{b}^{(m+1)}} &= \alpha \sum_{n=1}^N \widehat{\Q}_{\perp n}(t) + (1-\alpha)\widehat{Q}_{\perp m+1}(t)\\
    & \gp{a} \alpha \sum_{n=1}^N \widehat{\Q}_{\perp n}(t) + (1-\alpha)\widehat{Q}_{\perp m}(t)\\
    & = \inner{\widehat{\Q}_{\perp}(t)}{\mathbf{b}^{(m)}}\\
    & \ep{b} 0
  \end{align*}
  where (a) follows from Eq. \eqref{eq:qc_larger}; and (b) comes from Eq. \eqref{eq:i_in_I}. However, by Eq. \eqref{eq:all_i}, we must have $\inner{\widehat{\Q}_{\perp}(t)}{\mathbf{b}^{(m+1)}} \le 0$. Hence, Claim~\ref{claim_4} is true.
\end{proof}

\section{Proof of Proposition \ref{prop:prop_3}}
\label{sec:proof_prop_3}
\begin{proof}
  It follows from Lemma~\ref{claim_1} and the proof of Theorem~\ref{thm:theorem_1} that the key for heavy-traffic delay optimality is the term $\mathcal{T}^{(\epsilon)}$. Instead of using Cauchy-Schwartz inequality, we apply H\"older inequality in Eq. \eqref{eq:upper_T} to obtain a tighter bound of $\mathcal{T}^{(\epsilon)}$ as follows. 
  \begin{align}
  \label{eq:upper_T_general}
    \mathcal{T}^{(\epsilon)} & \le \ex{\inner{\overline{\UU}}{-N_1\overline{\Q}^+_{\perp}}}\nonumber\\
    & \lep{a} \frac{N}{\alpha^{(\epsilon)}} { \left(\ex{\norm{\overline{\UU}}^{r'}_{r'} } \right)^{\frac{1}{r'}}  \left(\ex{\norm{\overline{\Q}^+_{\perp}}^r_r} \right)^{\frac{1}{r}}}.\nonumber\\
    & \lep{b} \frac{N}{\alpha^{(\epsilon)}}  \left(c_{r'} \epsilon \right)^{\frac{1}{r'}} \left(\ex{\norm{\overline{\Q}^+_{\perp}}^r_2} \right)^{\frac{1}{r}}.\nonumber\\
    & \lep{c}\frac{N}{\alpha^{(\epsilon)}}  \left(c_{r'} \epsilon \right)^{\frac{1}{r'}} \left(\ex{\norm{\overline{\Q}_{\perp}}^r_2} \right)^{\frac{1}{r}}.
  \end{align}
  where (a) follows from H\"older inequality for random vectors, and $r, r'\in (1,\infty)$ satisfy $1/r + 1/r' = 1$; (b) comes from Lemma \ref{lem:unused_service} and the fact that if $0<r_1 < r_2$, then $\norm{\mathbf{x}}_{r_2} \le \norm{\mathbf{x}}_{r_1}$ holds for any vector $\mathbf{x}$; (c) is true since the distribution of $\Q(t+1)$ and $\Q(t)$ are the same in steady-state.

  Thus, in order to prove the result in Proposition \ref{prop:prop_3}, we are left to characterize the moment of $\overline{\Q}_{\perp}$ in terms of the parameter $\delta^{(\epsilon)}$. First, combining Eq. \eqref{eq:drift_in_delta} and Claim \ref{claim_3}, yields
  \begin{equation*}
    \begin{split}
      &\ex{\Delta V_\perp(\Q) \mid \Q(t) = \Q }\\
       \le & \frac{1}{2\norm{\Qc(t)}} \ex{\left(2\inner{\Qc(t)}{\A(t) - \s(t)} + L\right) \mid \Q(t) = \Q}\\
       \le & -\frac{\mu_{\Sigma}\delta}{4N}  + \frac{L}{2\norm{\Qc(t)}}\\
       \le & -\frac{\mu_{\Sigma}\delta}{8N} \text{ for all } \Q \text{ such that } \norm{\Qc} \ge \frac{4NL}{\mu_{\Sigma}\delta}.
    \end{split}
  \end{equation*}
  Thus the condition (C1) in Lemma \ref{lem:basis} is valid with $\eta = \frac{\mu_{\Sigma}\delta}{8N}$ and $\kappa = \frac{4NL}{\mu_{\Sigma}\delta}$. Also, from Eq. \eqref{eq:boundedQc}, we have the condition (C2) is valid with $D = \sqrt{N} \max(A_{\max},S_{\max})$. Then from Eq. \eqref{eq:upper_siva} in Lemma \ref{lem:basis}, we get for $r = 1,2,\ldots,$
    \begin{align}
    \label{eq:moment_in_delta}
            \ex{\norm{\overline{\Q}_{\perp}}^r_2} &\le (2\kappa)^r + (4D)^r\left(\frac{D+\eta}{\eta} \right)^r r!\nonumber\\
            &\le \frac{1}{\delta^r}K_r^r
    \end{align}
    where $K_r \triangleq \left[ \left(\frac{8NL}{\mu_{\Sigma}}\right)^r + r! \left(\frac{32D^2N + 4D\mu_{\Sigma}}{\mu_{\Sigma}}\right)^r \right]^{\frac{1}{r}}$, which is independent of $\epsilon$. Now, substituting Eq. \eqref{eq:moment_in_delta} into Eq. \eqref{eq:upper_T_general}, yields
    \begin{align*}
      \mathcal{T}^{(\epsilon)} \le F_r\frac{\epsilon^{(1-1/r)}}{\alpha^{(\epsilon)}\delta^{(\epsilon)}}
    \end{align*}
    where $F_r \triangleq NK_r(S_{\max})^{(1/r^2 - 1/r)}$, which is independent of $\epsilon$. Since $\alpha^{(\epsilon)}\delta^{(\epsilon)} = \Omega(\epsilon^{\beta})$ and $\beta \in [0,1)$, there exists a positive $\beta'$ such that 
    \begin{align*}
      \mathcal{T}^{(\epsilon)} = O(\epsilon^{\beta'}),
    \end{align*}
    when $r > 1/(1-\beta)$. This directly implies that $\lim_{\epsilon \to 0}\mathcal{T}^{(\epsilon)} = 0$. Thus from Lemma \ref{claim_1}, the given policy is heavy-traffic delay optimal.
\end{proof}

\end{document}